\documentclass{article}

\usepackage[numbers]{natbib}
    \usepackage[preprint]{neurips_2022}

\usepackage[utf8]{inputenc} %
\usepackage[T1]{fontenc}    %
\usepackage{hyperref}       %
\usepackage{url}            %
\usepackage{booktabs}       %
\usepackage{amsfonts}       %
\usepackage{nicefrac}       %
\usepackage{microtype}      %
\usepackage{xcolor}         %

\usepackage{booktabs} %
\usepackage[ruled,linesnumbered]{algorithm2e}

\SetKw{Continue}{continue}
\SetKw{Break}{break}
\SetAlFnt{\small}
\SetAlCapFnt{\small}
\SetAlCapNameFnt{\small}
\SetAlCapHSkip{0pt}
\IncMargin{-\parindent}
\usepackage{amsthm, amsmath,amsfonts,graphicx, amssymb,bm,geometry, mathtools}
\usepackage{hyperref}
\usepackage{thmtools}
\usepackage{cleveref}
\usepackage{float}
\usepackage{mathdots}
\usepackage{xcolor}
\usepackage{ifthen}

\usepackage{amsmath,amsfonts,bm}

\def\eqref#1{equation~\ref{#1}}

\def\1{\bm{1}}

\def\rx{{X}}
\def\ry{{Y}}

\def\rvx{{\mathbf{X}}}
\def\rvy{{\mathbf{Y}}}

\def\vx{{\mathbf{x}}}
\def\vy{{\mathbf{y}}}

\DeclareMathAlphabet{\mathsfit}{\encodingdefault}{\sfdefault}{m}{sl}
\SetMathAlphabet{\mathsfit}{bold}{\encodingdefault}{\sfdefault}{bx}{n}

\newcommand{\E}{\mathbb{E}}

\DeclareMathOperator*{\argmax}{arg\,max}

\theoremstyle{plain}
\newtheorem{theorem}{Theorem}[section]

\newtheorem{lemma}[theorem]{Lemma}

\theoremstyle{definition}
\newtheorem{definition}[theorem]{Definition}
\newtheorem{remark}[theorem]{Remark}

\newtheorem{assumption}[theorem]{Assumption}

\crefname{assumption}{assumption}{assumptions}

\newcommand{\compilehidecomments}{false}

\newcommand{\opt}[0]{{\text{opt} }}

\newcommand{\N}[0]{{\mathbb{N} }}

\usepackage{xcolor}

\title{Peer Prediction for Learning Agents}

\author{%
  Shi Feng \\
  Institute for Interdisciplinary Information Sciences,
  Tsinghua University\\
  Beijing, China \\
  \texttt{shifeng-thu@outlook.com} \\
  \AND 
  Fang-Yi Yu \\
  Department of Computer Science, George Mason University \\
  Fairfax, VA, USA \\
  \texttt{fangyiyu@gmu.edu}
  \AND
  Yiling Chen \\
  John A. Paulson School of Engineering and Applied Sciences, Harvard University \\
  Cambridge, MA, USA \\
  \texttt{yiling@seas.harvard.edu}
}

\begin{document}

\ifthenelse{ \equal{\compilehidecomments}{false} }{%
	\newcommand{\fang}[1]{}
	\newcommand{\shi}[1]{}
	\newcommand{\yiling}[1]{}
}{
	\newcommand{\fang}[1]{{\color{blue}  [\text{Fang-Yi:} #1]}}
	\newcommand{\shi}[1]{{\color{yellow!80!black} [\text{Shi:} #1]}}
	\newcommand{\yiling}[1]{{\color{red} [\text{Yiling:} #1]}}
}

\newcommand{\compilefullversion}{false}%
\ifthenelse{\equal{\compilefullversion}{false}}{%
	\newcommand{\OnlyInFull}[1]{}
	\newcommand{\OnlyInShort}[1]{#1}
}{%
	\newcommand{\OnlyInFull}[1]{#1}%
	\newcommand{\OnlyInShort}[1]{}%
}%

\maketitle

\begin{abstract}

Peer prediction refers to a collection of mechanisms for eliciting information from human agents when direct verification of the obtained information is unavailable. They are designed to have a game-theoretic equilibrium where everyone reveals their private information truthfully. This result holds under the assumption that agents are Bayesian and they each adopt a fixed strategy across all tasks. Human agents however are observed in many domains to exhibit learning behavior in sequential settings. In this paper, we explore the dynamics of sequential peer prediction mechanisms when participants are learning agents. We first show that the notion of no regret alone for the agents' learning algorithms cannot guarantee convergence to the truthful strategy. We then focus on a family of learning algorithms where strategy updates only depend on agents' cumulative rewards and prove that agents' strategies in the popular Correlated Agreement (CA) mechanism converge to truthful reporting when they use algorithms from this family satisfying an additional condition on full exploitation. This family of algorithms is not necessarily no-regret, but includes fixed-parameter versions of familiar online learning algorithms (e.g multiplicative weight update and Follow the Perturbed Leader) as special cases. Simulation of several algorithms in this family as well as the $\epsilon$-greedy algorithm, which is outside of this family, illustrates their finite-horizon behavior in the CA mechanism.

\end{abstract}

\section{Introduction}

A fundamental challenge in many domains is to elicit high-quality information from people when directly verifying the acquired information is {\em not} feasible, either because the ground truth is not available or because it's too costly to obtain. Notable settings include asking people to label data for machine learning, having students perform peer grading in education, and soliciting customer feedback for products and services. 

The peer prediction literature has made impressive progress on this challenge in the past two decades, with many mechanisms that have desirable incentive properties developed for this problem \cite{prelec2004bayesian, miller2005eliciting, witkowski2012robust, dasgupta2013crowdsourced, radanovic2014incentives, faltings2014incentives, kamble2015truth, shnayder2016informed,kong2016framework, kong2016equilibrium, prelec2017solution, liu2018surrogate, schoenebeck2020learning,kong2020dominantly}. The term peer prediction refers to a collection of reward mechanisms that solicit information from human agents and reward each agent solely based on how the agent's reported information compares with that of the other agents, without having access to the ground truth. Under some assumptions, many peer prediction mechanisms \cite{prelec2004bayesian, miller2005eliciting, witkowski2012robust, radanovic2014incentives, faltings2014incentives, kong2016equilibrium, prelec2017solution} guarantee that every agent truthfully reporting their information is a game-theoretic equilibrium, and the more recent multi-task peer prediction mechanisms \cite{dasgupta2013crowdsourced, kamble2015truth, shnayder2016informed, kong2016framework, liu2018surrogate, schoenebeck2020learning, kong2020dominantly} further ensure that agents receive the highest expected payoff at the truthful equilibrium, compared with other strategy profiles. 

While achieving truthful reporting as a highest-payoff equilibrium is a victory to declare for this challenging without-verification setting, there are however caveats associated with adopting the notion of equilibrium as a solution concept. The equilibrium results rely on the assumption that participants are fully rational Bayesian agents. Equilibrium is a static notion and doesn't address how agents, who act independently, jump to play their equilibrium strategies. Moreover, the equilibrium results of multi-task peer prediction mechanisms heavily depend on a {\em consistent strategy} assumption, that is, each agent is assumed to adopt a fixed strategy across all tasks that she participates. All together these assumptions exclude the possibility that agents may explore and learn from previous experience, a behavior that's not only commonly observed in practice but also has been modeled in studying other strategic settings \cite{braverman2017selling, deng2019strategizing, camara2020mechanisms}. %

This paper is the first theoretical study on the dynamics of sequential peer prediction mechanisms when participants are learning agents. The main question that we explore is whether and when in sequential peer prediction, learning agents will converge to all playing the truthful reporting strategy. We first consider agents adopting no-regret learning algorithms and prove that the notion of no regret alone cannot guarantee convergence to truthful reporting. We then define a natural family of reward-based learning algorithms where strategy updates only depend on agents' cumulative rewards. While algorithms in this family is not necessarily no-regret (e.g. the Follow the Leader algorithm), this family includes fixed-parameter versions of familiar online learning algorithms, including the Multiplicative Weight Update and the Follow the Perturbed Leader algorithms. Our main result shows that, for the binary-signal setting, agents' strategies in the popular Correlated Agreement (CA) mechanism \cite{dasgupta2013crowdsourced} converge to truthful reporting when agents use algorithms from this family satisfying the additional condition on full exploitation in \cref{thm.converge}. %
To prove the result, we show the process has a self-fulfilling property: once Alice and Bob have large accumulated rewards for truth-telling, they are more likely to play truth-telling and resulting in larger accumulated rewards.  Theoretically, we carefully partition the process into three stages, bad, intermediate, and good events illustrated in \cref{Fig.schemfig}, and use tools in martingale theory to argue the progress of the process.  Finally, we simulate the strategy dynamics in the CA mechanism for several algorithms in this family as well as for the $\epsilon$-greedy algorithm, which doesn't belong to this family. We present finite-horizon simulations of these algorithms, suggesting an interesting future direction to characterize all learning algorithms that converge to truthful reporting.

\paragraph{Related Works}
This paper relates to two lines of work, information elicitation and mechanisms for learning agents. 
  
\textbf{Information Elicitation Mechanisms}  The literature on information elicitation without verification focuses on capturing the strategic aspect of human agents.  In multi-task settings,  \citet{dasgupta2013crowdsourced} proposed a seminal informed truthful mechanism, the Correlated Agreement (CA) mechanism, for binary positively correlated signals.  A series of works then relaxed the binary and positively correlation assumptions~\cite{shnayder2016informed,kong2016framework,schoenebeck2020learning,kong2020dominantly}. Additionally, \citet{DBLP:journals/corr/abs-2106-03176} study the limitation of information elicitation in the multi-task setting. However, all of the above works assume agents using consistent strategies that are identical across all tasks.  The consistent strategy assumption excludes the possibility of agent learning.  %
Our work removes the consistent strategy assumption and explicitly considers learning agents. We theoretically prove truthful convergence of the CA mechanism when agents use reward-based online learning algorithms satisfying an additional condition on full exploitation.

Our work of considering learning agents can be viewed as a way of testing the robustness of information elicitation mechanisms with respect to deviation from the rational Bayesian agent model. From this perspective, \citet{shnayder2016measuring} is closely related to ours. They consider sequential information elicitation and empirically study if agents using replicator dynamics can converge to truth-telling in the CA mechanism and several other mechanisms.  Our work theoretically proves that, besides replicator dynamics, learning agents can converge to truth-telling in the CA mechanism when they use reward-based online learning algorithms satisfying the additional condition on full exploitation.  Additionally,  \citet{schoenebeck2021information} designed an information elicitation mechanism that was robust against a small fraction of adversarial agents.

\textbf{Mechanisms for Learning Agents} Several works in economics and computer science try to design mechanisms for learning agents, rather than for rational, Bayesian ones.  
\citet{braverman2017selling} studied pricing mechanisms for learning agents with no external regrets called mean-based algorithm.\fang{} Their work was generalized by \citet{deng2019strategizing} to consider repeated Stackelberg games in full-information settings.  \citet{deng2022nash} studied auction for mean-based algorithm and show the convergence to Nash equilibrium.  \citet{camara2020mechanisms} further proposed counterfactual internal regrets (CIR) together with no-CIR assumption, which was proved to be a sufficient behavior assumption for no-regret principal mechanism design in repeated stage games. However, all of these works focus on a single agent or full-information games, while peer prediction is an incomplete-information game with multiple agents.
Finally, our goal is slightly different from that of most sequential mechanism design.  Instead of maximizing the mechanism designer's utility, our goal is to incentivize truthful reporting from agents.

\section{Peer Prediction Settings}\label{sec:pre}

For simplicity we consider two agents, Alice and Bob, who work on a sequence of tasks indexed by $t\ge 1$.\footnote{For an even number of agents, the results apply separately to disjoint pairs satisfying the assumptions; different pairs may converge to different options.}  For round $t$, both agents work on task $t$, Alice receives a signal $\rx_t = x_t$ in $\{0,1\}$, and Bob a signal $\ry_t = y_t$ in $\{0,1\}$ where $\rx_t$ and $\ry_t$ denote random variables, and $x_t$ and $y_t$ are their realizations.  Then Alice and Bob report $\hat{\rx}_t = \hat{x}_t$ and $\hat{\ry}_t = \hat{y}_t$ in $\{0,1\}$.
We define $\rvx_{\le t} = \{\rx_s: 1\le s\le t\}\in \{0,1\}^t$ and $\hat{\rvx}_{\le t} = \{\hat{\rx}_s: 1\le s \le t\}\in \{0,1\}^t$ to denote Alice's {signal profile} and report profiles until $t$-th round respectively, and define $\rvy_{\le t}$ and $\hat{\rvy}_{\le t}$ for Bob similarly.  We use $\rvx, \hat{\rvx}, \rvy$, and $\hat{\rvy}$ for the complete signal and report profiles. 
Additionally, we consider the signals are generated from some distribution $\mathbb{P}$ that satisfies the following assumptions:
\begin{assumption}[name = A priori similar tasks~\cite{dasgupta2013crowdsourced}, label = asm.apriori]
Each pair of signal is identically and independently (i.i.d.) generated: there exists a distribution $P_{X,Y}$ over $\{0,1\}^2$ such that $(X_t,Y_t)\sim P_{X,Y}$ for any $t\in\mathbb{N}^+$. Moreover, we assume the distribution has full support, $P_{X,Y}(x,y)>0$ for all $x,y\in\{0,1\}$.
\end{assumption}

\begin{assumption}[name = Positively correlated signals, label = asm.poscorr]
The distribution $P_{X,Y}$ is positively correlated, $\min\{P_{X,Y}(1,1),P_{X,Y}(0,0)\}>\max\{P_{X,Y}(1,0),P_{X,Y}(0,1)\}.$
\end{assumption}

Now we introduce multi-task peer prediction mechanisms and sequential peer prediction mechanisms, and their relation.  We will focus on the sequential setting.
Multi-task peer prediction mechanisms work on a fixed number of tasks.  Formally, a multi-task peer prediction mechanism on $k$ tasks is a pair of payment functions $\bar{M}: \{0,1\}^{2\times k}\to [-1,1]^2$.  For instance, the \emph{(multi-task) correlated agreement mechanism} (CA mechanism)\footnote{While the CA mechanism can be defined on non binary setting and does not require positive correlation.~\cite{shnayder2016informed}, with \cref{asm.poscorr}, the CA mechanism reduces to \cref{eq:ca} and is first proposed in \cite{dasgupta2013crowdsourced}.  Finally, when the number of task is greater than two, we can compute the payment based on the last two tasks or two random tasks since agents using consistent strategy and \cref{asm.apriori}.}~\cite{dasgupta2013crowdsourced,shnayder2016informed,shnayder2016measuring} is $\bar{M}^{CA}(\hat{\vx}, \hat{\vy}) = \left(\mathbb{I}[\hat{x}_2=\hat{y}_2]-\mathbb{I}[\hat{x}_2=\hat{y}_{1}], \mathbb{I}[\hat{y}_2=\hat{x}_2]-\mathbb{I}[\hat{y}_2=\hat{x}_{1}]\right)$ for all $\hat{\vx}, \hat{\vy}\in \{0,1\}^2$.  Intuitively, the CA mechanism rewards agreement on the same task and punishes agreement on uncorrelated tasks.

A \emph{sequential information elicitation mechanism} is a sequence of payment functions $\mathcal{M} = \{M_t:t\ge 1\}$ where $M_t:\{0,1\}^{2\times t}\to [-1,1]^2$ for all $t$.  After Alice and Bob reporting $\hat{\vx}_{\le t}$ and $\hat{\vy}_{\le t}$ in round $t$, the mechanism computes $(r_t, s_t) := M_{t}(\hat{\vx}_{\le t}, \hat{\vy}_{\le t})$ and pay $r_t$ to Alice and $s_t$ to Bob.  Here we assume $M_t$ can only depends on a constant $k$ round of reports so that $M_{t}(\hat{\vx}_{\le t}, \hat{\vy}_{\le t}) = M_{t}(\hat{x}_{t-k+1},\dots,\hat{x}_{t}, \hat{y}_{t-k+1},\dots,\hat{y}_{t})$ for all $t\ge k$, $\hat{\vx}_{\le t}$, and $\hat{\vy}_{\le t}$, where $k$ is a fixed positive integer; we call such $\mathcal{M}$ a \emph{rank $k$ mechanism}. For $t<k$, payments are fixed bounded functions of the available reports.
For instance, the (multi-task) CA mechanism can be adopted as a sequential rank $2$ information elicitation mechanism:  At round $t$, the payment is 
\begin{equation}\label{eq:ca}
    M_t^{CA}(\hat{\vx}_{\le t}, \hat{\vy}_{\le t}) = \left(\mathbb{I}[\hat{x}_t=\hat{y}_t]-\mathbb{I}[\hat{x}_t=\hat{y}_{t-1}], \mathbb{I}[\hat{y}_t=\hat{x}_t]-\mathbb{I}[\hat{y}_t=\hat{x}_{t-1}]\right)
\end{equation}
where $\hat{x}_0$ and $\hat{y_0}$ are set as $0$.  
Similarly, we say a sequential information elicitation mechanism $\mathcal{M} = (M_t)_{t\ge 1}$ is a \emph{sequential version} of a multi-task information elicitation mechanism $\bar{M}$ if $M_t(\hat{\vx}_{\le t}, \hat{\vy}_{\le t}) = \bar{M}(\hat{x}_{t-k+1},\dots,\hat{x}_{t}, \hat{y}_{t-k+1},\dots,\hat{y}_{t})$
for all $\hat{\vx}_{\le t}, \hat{\vy}_{\le t}$ and $t\ge k$.  The payment at round $t\ge k$ is $\bar{M}$ on the latest $k$ reports.  
Conversely, a sequential information elicitation mechanism can be seen as a sequence of multi-task information elicitation mechanisms.
Now we formally define agents' strategies. Due to symmetry, we introduce notation for Alice and omit Bob's.  Given an information elicitation mechanism $\mathcal{M}$, at round $t$, Alice observes her signal $x_t$ and decides on her report $\hat{x}_t$.  Thus, Alice has four \emph{options (pure strategies)}: 1) $\opt_1$: report the private signal truthfully, 2) $\opt_2$: flip the private signal, 3) $\opt_3$: report $1$ regardless of the signal, and 4) $\opt_4$: report $0$ regardless of the signal.  We call $\opt_3$ and $\opt_4$ uninformative strategies.  We use $\opt^X_t$ to denote Alice's pure strategy, and $r_t$ for her payoff at round $t$.  At each round $t$, Alice knows her previous signals $\vx_{\le t}\in \{0,1\}^t$, her pure strategies $\opt^X_1, \dots, \opt^X_{t-1}$, and Bob's reports $\hat{\vy}_{\le t-1}$, so we use $\mathcal{F}_t = \{\vx_{\le t}, \opt^X_{\le t-1}, \hat{\vy}_{\le t-1}\}$
to denote Alice knowledge at round $t$.  Thus, Alice's mixed strategy at round $t$ is a stochastic mapping $\sigma^X_t$ from $\mathcal{F}_t$ to $\{\opt_1, \opt_2, \opt_3, \opt_4\}$.  We'll abuse our nation and also use $\sigma_t^X = \opt^X_t$ to represent the realized pure strategy.\fang{}  Finally, a learning algorithm of Alice is a mapping from an information elicitation mechanism $\mathcal{M}$ to her strategies. 

\yiling{}\yiling{}
\fang{}

\paragraph{Strongly Truthful for Rational and Bayesian agents} 
Previous works on information elicitation try to ensure truth-telling $\opt_1$ is the best strategy for rational and strategic agents.  In particular, a mechanism is \emph{strongly truthful} if  Alice and Bob report truthfully is a Bayesian Nash Equilibrium (BNE) and they get strictly higher payment at this BNE than at any other non-permutation BNE.  We present the formal definitions in the appendix. Informally, in a permutation BNE, every agent's strategy on each round is a permutation/bijection from his/her signals to reports.\fang{}
However, the equilibrium results of previous mechanisms not only require \cref{asm.apriori} but further assume agents using \emph{consistent strategies}.  Specifically, Alice uses a consistent strategy if there is a fixed distribution on $\{\opt_1, \opt_2, \opt_3, \opt_4\}$ so that $\opt^X_t$ is generated from a fixed distribution that is independent of her private signals on other tasks and the round number.\yiling{}\fang{}   
For instance, when Alice and Bob are Bayesian and use consistent strategies under \cref{asm.poscorr} and \ref{asm.apriori}, \citet{dasgupta2013crowdsourced} show CA mechanism in \cref{eq:ca} is strongly truthful.    Intuitively, positive correlation \cref{asm.poscorr} guarantees that truthful reporting can maximize the chance of agreeing with the peer on the same task while avoiding agreeing on reports on other tasks.  Furthermore, in \cref{app.bayesiantruthful} we show CA mechanism merely has three types Bayesian Nash equilibria, at which both agents 1) play truth-telling $\opt_1$, 2) flip the signal $\opt_2$, or 3) generate uninformative reports (mixture between $\opt_3, \opt_4$) when agents use consistent strategies and \cref{asm.poscorr,asm.apriori} hold.

\yiling{}\shi{}\fang{}

\paragraph{Truthful convergence for learning agents}
However, as we consider agents using a family of online learning algorithms to decide their strategies, standard solution concepts like Bayesian Nash equilibrium no longer apply.  Additionally, online learning algorithms often have exploration, so we cannot hope agents will always use the truth-telling strategy.  For learning agents, our goal is to test whether existing  mechanisms can ensure that agents will \emph{converge} to truthful reporting when they deploy certain learning behavior that goes beyond obliviously consistent strategies.\yiling{}  \fang{}

We now formalize the convergence of algorithms to truthful reporting.  Because we want to elicit information without verification, it is information-theoretically impossible for us to separate permutation equilibrium, where all agents play $\opt_2$, from truthful equilibrium, where all agents play $\opt_1$, without any additional information~\cite{kong2016framework}. 
However, if we have an additional bit of information on whether the prior of $0$ is larger than $1$, we may tell apart these two equilibria. \fang{}We hence define convergence to truthful reporting as the limits of both $\opt^X_t$ and $\opt^Y_t$ being truth-telling ($\opt_1$) or flipping ($\opt_2$).  
Note that \cref{def.truthconvergence} requires almost surely convergence which is very strong convergence concept.  
\begin{definition}\label{def.truthconvergence}
An information elicitation mechanism $\mathcal{M}$ achieves \emph{truthful convergence} for agents using algorithms $A_1$ and $A_2$ respectively if and only if both sequences of pure strategies converge to truth-telling or both flipping the reports.
$$\Pr\left\{\lim_{t\rightarrow +\infty}\opt^X_t=\lim_{t\rightarrow +\infty}\opt^Y_t=\opt_1\vee\lim_{t\rightarrow +\infty}\opt^X_t=\lim_{t\rightarrow +\infty}\opt^Y_t=\opt_2\right\}=1.$$
\end{definition}
\yiling{}\shi{}\fang{}

\section{Online Learning Algorithms}
In this section, we explore candidates to model agents' learning behavior to replace Bayesian agents' consistent strategies in the literature.  
We first examine the relationship between no regret and truthful convergence in \cref{sec:noregret}.
Then in \cref{sec:rewardbased}, we introduce a family of reward-based online learning algorithm to model agents' learning behavior, and show that the family of reward-based online learning algorithms contains fixed-parameter versions of several common online learning algorithms as special cases.

\subsection{No-regret online learning algorithms}\label{sec:noregret}
We now investigate the relationship between no regret and truthful convergence.  First, we show general no-regret algorithms may not ensure truthful convergence (\cref{thm:impossible}).  However, we show the converse is almost true (\cref{thm:converge2noregret}): If truthful convergence happens, the agents do not have regret when the sequential mechanism is a sequential version of a strongly truthful mechanism whose common-flipping profile is also a Bayesian Nash equilibrium.  \fang{}

Given a sequential information elicitation mechanism $\mathcal{M} = (M^X_t, M^Y_t)_{t\ge 1}$, signals $\vx, \vy$, and reports $\hat{\vx},\hat{\vy}$, we define $r_{i,t} = M_t^X(\opt_i(x_1),\dots,\opt_i(x_t),\hat{\vy}_{\le t})$ be the payoff when Alice uses strategy $\opt_i$ and Bob's choices are unchanged.  Then Alice's \emph{regret} is 
$Reg^X(T) = \max_i \sum_{t\le T} r_{i,t}-\sum_{t\le T} r_t$.  Finally, we say that Alice's and Bob's online learning algorithms are \emph{no regret} (on $\mathcal{M}$) if $\E[Reg^X(T)] = \E[Reg^Y(T)] = o(T)$ over the randomness of signals and the algorithms, and we say Alice and Bob are no regret for short. \yiling{} 

One may hope that no regret as a behavior assumption for agents is sufficient for achieving  desirable outcome in a mechanism.  
However, the following theorem shows that we cannot have an information elicitation mechanism that achieves truthful convergence for all no-regret agents.
\begin{theorem}[label = thm:impossible]
For any sequential information elicitation mechanism $\mathcal{M}$ of rank $k\in \mathbb{N}^+$, there exist no-regret algorithms for Alice and Bob so that $\mathcal{M}$ cannot achieve truthful convergence.
\end{theorem}
The main idea of the proof is that the no-regret assumption cannot prevent Alice and Bob from colluding. In our counterexample, Alice and Bob decide on a no-regret sequence of reports regardless of their signals once the mechanism is announced. 
Technically, we use probabilistic method to show the existence of a deterministic and no-regret sequence of strategies $(\opt^X_t, \opt^Y_t)_{t \ge 1}$ that consists of reporting $1$ or $0$ regardless of private signal, i.e. $\opt_3$ or $\opt_4$.  The formal proof is in \cref{app:impossible}.

The notion of truthful convergence in \cref{def.truthconvergence} provides an ideal truthful guarantee to the mechanism designer.  Here we show that truthful convergence also ensures no regret for agents when the sequential mechanism is a sequential version of a strongly truthful one-shot multi-task mechanism and the common-flipping profile is also a Bayesian Nash equilibrium. \yiling{}\fang{}

For instance, if a pair of algorithms exhibits truthful convergence on the sequential CA mechanism (\cref{eq:ca}), they are also no regret (on the game). Intuitively, if Bob converges to the truth-telling $\lim_{t\to \infty} \opt^Y_t = \opt_1$, the average expected gain of Alice deviating to $\opt_i$ is equal to the expected gain of deviating to $\opt_i$ when Bob always tells the truth.
The gain is non positive because the CA mechanism is strictly truthful by~\cref{thm.strictne}.Therefore, the expected regrets $\E[Reg^X(T)]$ and $\E[Reg^Y(T)]$ are small.
\Cref{thm:converge2noregret} formalizes and extends the above idea to strongly truthful multi-task information elicitation mechanisms whose common-flipping profile is also a Bayesian Nash equilibrium.
\yiling{}

\begin{theorem}[label = thm:converge2noregret]
Let $\bar{M}$ be a strongly truthful multi-task information elicitation mechanism, and $\mathcal{M}$ be a sequential version of $\bar{M}$. Assume also that the common-flipping profile $(\opt_2,\opt_2)$ is a Bayesian Nash equilibrium of $\bar{M}$. If $\mathcal{M}$ achieve truthful convergence for Alice and Bob using algorithms $A_1$ and $A_2$ respectively, then Alice and Bob are no regret.
\end{theorem}

\subsection{Reward-based Online Learning Algorithms}\label{sec:rewardbased}As shown in \cref{thm:impossible}, the no-regret assumption allow does not guarantee truthful convergence.  
In this section, we introduce a general family of online learning algorithms, \emph{reward-based online learning algorithm}, under the general full feedback bandit setting, and we will apply these algorithms in sequential information elicitation mechanisms later.\fang{}  Informally, a reward-based online learning algorithm decides each round's strategy using a fixed update function that depends only on the accumulative reward.

For simplicity, we only consider algorithms on four strategies $\{\opt_1, \opt_2, \opt_3, \opt_4\}$ which can be extended easily.   Recall that the payoff of choosing option $\opt_i,i\in[4]$ at round $t$ is $r_{i,t}$ when others' choices are unchanged. We denote the accumulated payoffs of these four options as $R_{i,t}=\sum_{j=1}^tr_{i,j}$ for $i\in[4]$.
Symmetrically, we use $S_{i,t}$ and $s_{i,t}$ to represent accumulated payoffs and the payoff of turning to choose option $\opt_i$ in the $t^{th}$ round for Bob.  For example, for our specific peer prediction game using CA mechanism in \cref{eq:ca}, $r_{1,t}=\mathbb{I}[x_t=\hat{y}_t]-\mathbb{I}[x_t=\hat{y}_{t-1}]$ and therefore, $R_{1,t}=\sum_{j=1}^t(\mathbb{I}[x_j=\hat{y}_j]-\mathbb{I}[x_j=\hat{y}_{j-1}])$ for Alice.

\fang{}

We consider a family of \emph{reward-based online learning algorithms} $\mathcal{A}$ that use an \emph{update function} $f:\mathbb{R}^4\rightarrow \triangle^3$, and choose $\opt_i$ with probability $f_i(R_{1,t-1},R_{2,t-1},R_{3,t-1},R_{4,t-1})$ for $i\in[4]$ in the $t^{th}$ round, where $f_i$ is the $i^{th}$ coordinate of $f$. A such mechanism based on accumulated payoffs is denoted by $A_f$. We have three assumptions for the update function $f$, which are all very natural. The first two require that $f$ is exchangeable
and preserves ordering.
\begin{assumption}[name=Exchangeability of $f$, label=asm.symf]
For any $R_1,R_2,R_3,R_4\in\mathbb{R}$ and an arbitrary permutation of them $R_{i_1},R_{i_2},R_{i_3},R_{i_4}$,  $f_{i_j}(R_1,R_2,R_3,R_4)=f_j(R_{i_1},R_{i_2},R_{i_3},R_{i_4})$ for all  $j\in[4]$.
\end{assumption}

\begin{assumption}[name=Order preservation of $f$, label=asm.consistf]
For any $R_1,R_2,R_3,R_4\in\mathbb{R}$ and suppose that $R_{i_1},R_{i_2},R_{i_3},R_{i_4}$ is a non-increasing order of them, for $f$ we have $
    f_{i_1}(R_1,R_2,R_3,R_4)\geq f_{i_2}(R_1,R_2,R_3,R_4)\geq f_{i_3}(R_1,R_2,R_3,R_4)\geq f_{i_4}(R_1,R_2,R_3,R_4)$. 
\end{assumption}

Finally we consider the strategy chosen by the update function $f$ when the accumulated payoff of an strategy is much higher than that of other strategies (\cref{asm.fullsupf}).  Appendix~\ref{app.justasm} shows that the assumption is necessary for no regret in general online decision problems when the displayed limit exists.

\begin{assumption}[name=Full exploitation of $f$, label=asm.fullsupf] $\lim_{R_1-\max\{R_2,R_3,R_4\}\rightarrow+\infty}f_1(R_1,R_2,R_3,R_4)=1$.
\end{assumption}

Now we show that the family of reward-based online learning algorithms $\mathcal{A}$ satisfying \cref{asm.symf,asm.fullsupf,asm.consistf} contains fixed-parameter versions of several classic online learning algorithms \cite{freund1997decision, kalai2005efficient}.
\begin{theorem}[label = thm:specialcases]
$\mathcal{A}$ contains Follow the Perturbed Leader (FPL) algorithm, and Multiplicative Weights algorithm as special cases. Corresponding $f$'s for them are listed as below:\yiling{}\shi{}
\begin{description}
\item[Multiplicative Weight algorithm] $f_i^{\text{hedge 1}}(R_1,R_2,R_3,R_4)=\frac{e^{\beta R_i}}{\sum_{j\in[4]}e^{\beta R_j}}$ for $i\in [4]$ 
\item[FPL algorithm] Given i.i.d. noise variables $p_j\sim\mathcal{N}$, let $W(R,p)=\argmax_{j\in[4]}(R_j+p_j)$. With uniform tie-breaking, $f_i^{\text{FPL}}(R)=\mathbb{E}[\mathbb{I}[i\in W(R,p)]/|W(R,p)|]$ for $i\in[4]$.
\end{description}
\end{theorem}

In the binary peer prediction problem using CA mechanism, we can further show that replicator dynamics and linear updating multiplicative weight algorithm \cite{arora2012multiplicative} are both in $\mathcal{A}$ in \cref{app.morealgina}.

\begin{remark}
Our reward based online learning is very similar to the mean-based algorithm in \cite{braverman2017selling} and both use the accumulated rewards to characterize the algorithm's choice.  Additionally, like the mean-based algorithms, our reward based online algorithm may contain algorithms with regret in genral game, e.g., follow the leader.   
The family of reward based online algorithms uses an identical update function across all rounds.  Thus time-varying algorithms, e.g., $\epsilon$-greedy with time decreasing $\epsilon$, doesn't belong to family.  Finally, if a mean-based algorithm uses the same update function in each round and satisfies \cref{asm.symf,asm.fullsupf,asm.consistf}, then it belongs to this family.  Additionally, \fang{}then the update function needs to satisfy \cref{asm.fullsupf}.\yiling{}
\end{remark}

\section{Truthful Convergence of CA Mechanism on Reward-based Algorithms}
\label{sec.main}
Now we present our main result.  We will show that the sequential CA mechanism can achieve truthful convergence if both agents use reward-based online learning algorithms from $\mathcal{A}$ satisfying the additional condition below.  This convergence result suggests that the classical CA is robust even when agents deviate from Bayesian rational behavior and use a general family of online learning algorithms.

Conditional on the history through round $t-1$, the agents sample their pure strategies independently using $f$ and $g$, with fresh private randomness independent of the new signal pair $(X_t,Y_t)$. The i.i.d. signal stream is independent of this private randomness.

\begin{theorem}\label{thm.converge}
Under \cref{asm.apriori,asm.poscorr}, the binary-signal, sequential CA mechanism as defined in \cref{eq:ca} achieves truthful convergence when agents use reward-based algorithms $A_f$ and $A_g$, where the update functions $f$ and $g$ satisfy \cref{asm.symf,asm.fullsupf,asm.consistf} and the following additional condition on full exploitation: for each $k\in\{f,g\}$ there exists a nonincreasing $q_k:(0,\infty)\to[0,1]$ such that
\[
    \int_1^\infty q_k(x)\,dx<\infty,\qquad
    1-k_i(V)\leq q_k\!\left(V_i-\max_{j\neq i}V_j\right)
\]
for every $V\in\mathbb{R}^4$ and $i\in[4]$ for which the displayed accumulated payoff difference is positive.
\end{theorem}

The additional condition does not change \cref{asm.fullsupf}, the family $\mathcal{A}$, or \cref{def.truthconvergence}. The reward-process lemmas use only the original assumptions; the additional condition is used in the last step. Without it, the proof still gives a common truth-telling or flipping option whose choice probabilities tend to one almost surely.

Note that given agents' online learning algorithm and the payment function, the sequence of accumulated reward vector $(R_{1,t}, R_{2,t}, R_{3,t},R_{4,t},S_{1,t}, S_{2,t}, S_{3,t}, S_{4,t})_{t\ge 1}$ forms a stochastic process and we define $\mathcal{H}_t$ as the complete history of both agents' signals, pure strategies and reports through round $t$.  Additionally, if the accumulated reward of truth-telling $R_{1,t}, S_{1,t}$ is much larger than the others', full exploitation (\cref{asm.fullsupf}) makes both agents' probabilities of choosing truth-telling tend to one; the added condition will ensure convergence of the realized pure strategies. \yiling{} Thus, it is sufficient for us to track the evolution of accumulated reward vector.  Though the process of accumulated reward vector is not a Markov chain because the payment function~\cref{eq:ca} depends on reports in two rounds, we can still use ideas from semi-martingale to track the process.

Before proving our main result, we first present two properties of CA mechanism for our binary signal peer prediction problem. Lemma~\ref{lemma.boundeduninformative} shows that the accumulated payoffs of uninformative strategies $\opt_3,\opt_4$ $R_{3,t},R_{4,t},S_{3,t}$, and $S_{4,t}$  are always bounded. 
\begin{lemma}\label{lemma.boundeduninformative}
Given the game defined in \cref{thm.converge}, for any round $t$, the accumulated payoffs $R_{3,t},R_{4,t},S_{3,t}$, and $S_{4,t}$ for two agents are bounded by $[-1,1]$.
\end{lemma}

The second one, Lemma~\ref{lemma.sumofpayoffs}, tells us that the summation of accumulated payoffs of $\opt_1$ and $\opt_2$ is always fixed and is equal to the summation of accumulated payoffs of uninformative ones $\opt_3$ and $\opt_4$ for both agents. The proofs of these two lemmas are in \cref{app.propertyca}.
\begin{lemma}\label{lemma.sumofpayoffs}
Given the game defined in \cref{thm.converge}, for any round $t$, Alice has $R_{1,t}+R_{2,t}=R_{3,t}+R_{4,t}=0$, and Bob has $S_{1,t}+S_{2,t}=S_{3,t}+S_{4,t}=0$.
\end{lemma}

We now sketch the proof that consists of four steps. Informally, using \cref{lemma.boundeduninformative,lemma.sumofpayoffs}, the first step says that the uninformative strategies $\opt_3$ and $\opt_4$ can not completely dominate other strategies.  Specifically, both agents can not choose an option between $\opt_3$ and $\opt_4$ with a probability larger than $0.75$. 
With the first step, \cref{lemma.boundeduninformative} and \cref{lemma.sumofpayoffs}, we only need to focus on agents' reward of the truthtelling $\opt_1$ shown in \cref{Fig.schemfig}. We partition the space into three types of events.  \emph{Good events} happen when $R_1$, $S_1$ are both very large or very small, \emph{bad events} happen when one of $R_1,S_1$ is very large and one of them is very small, and \emph{intermediate states} are the states between them. The second step removes the possibility that Alice and Bob continue using different reports $\opt_1$ and $\opt_2$. Therefore, we can always escape "bad events" and enter "intermediate states". The third step further shows that if the game is in an intermediate state, there exists a constant probability that the game will get into a good events that leads to truthful convergence in a constant number of rounds. Hence, after the game enters "good events", which leaves us the final step: showing their strategies converges to either both truth telling $\opt_1$ or flipping $\opt_2$ truthful convergence.

Now we discuss each steps in more details but defer all the proofs to \cref{app.mainproof}.
\begin{figure}[htbp] 
\centering
\includegraphics[width=0.55\textwidth]{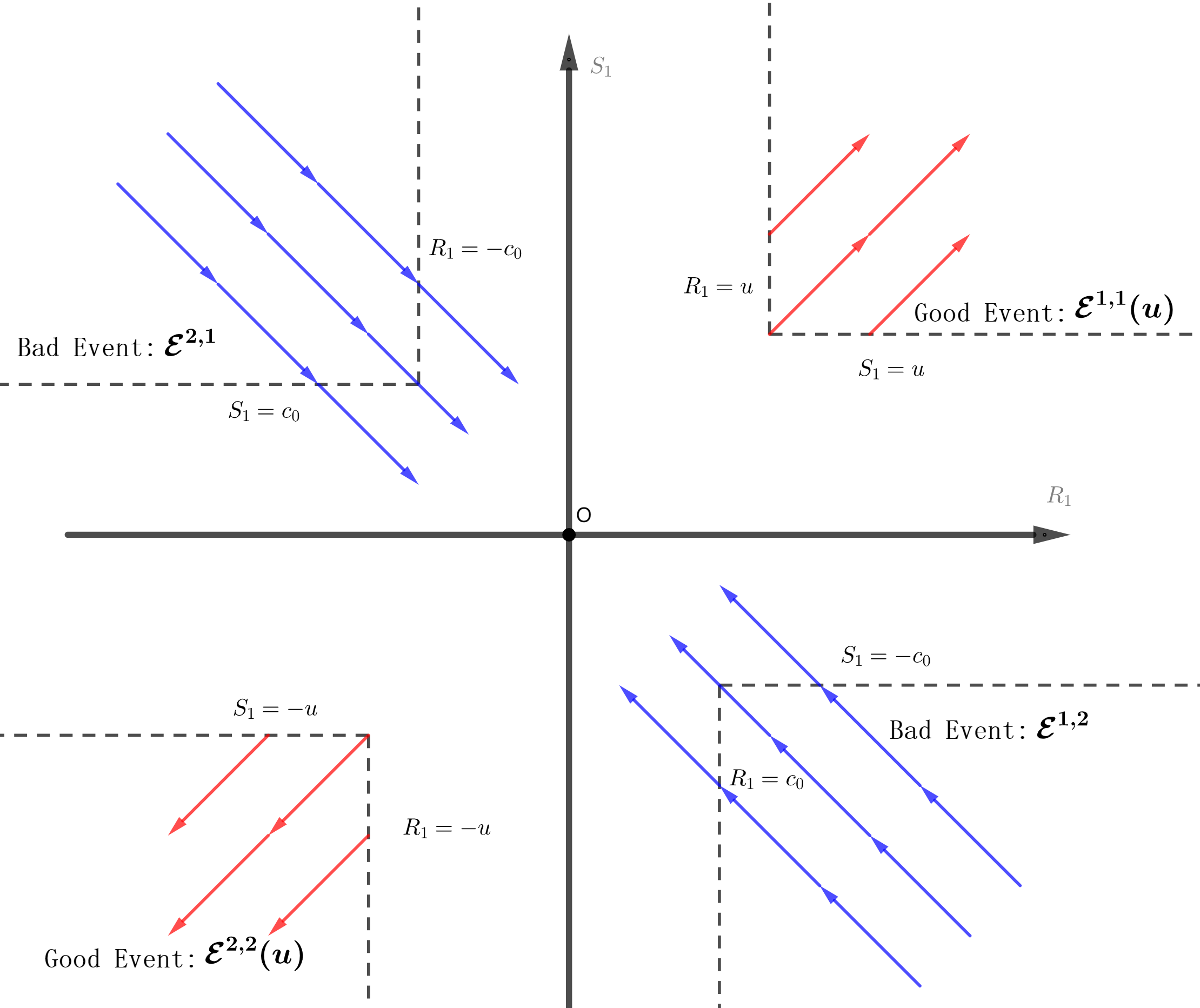}
\caption{A schematic diagram of behaviors of $R_{1,t}$ and $S_{1,t}$.}
\label{Fig.schemfig}
\end{figure}

\textbf{Step 1: Choosing \texorpdfstring{$\opt_1,\opt_2$}{Lg} with Nonzero Probability}

Combining Lemmas~\ref{lemma.boundeduninformative} and \ref{lemma.sumofpayoffs}, we have the following Lemma~\ref{lemma.nonoptimality34}, which completes step 1. 

\begin{lemma}\label{lemma.nonoptimality34}
Given the game defined in \cref{thm.converge}, for any round $t\ge 0$ and $i\in\{1,2\}$, if $\sum_{j=1}^tr_{i,j}\geq0$, the conditional probability for Alice to choose $\opt_i$ in round $t+1$ is at least $\frac{1}{4}$; if $\sum_{j=1}^ts_{i,j}\geq0$, the conditional probability for Bob to choose $\opt_i$ in round $t+1$ is at least $\frac{1}{4}$.
\end{lemma}

\textbf{Step 2: Escaping Bad Events}

Before introduce the formal statement of step 2, we define two "bad events".  Given $c_0$ for all $t\ge 1$
$\mathcal{E}_{t}^{1,2}:=\{R_{1,t}>c_0, \text{ and }S_{2,t}>c_0\}\text{ and }\mathcal{E}_{t}^{2,1}:=\{R_{2,t}>c_0\text{ and }S_{1,t}>c_0\}$.
We will specify $c_0$ later. Intuitively, when $c_0$ is sufficiently large, $\mathcal{E}_{t}^{1,2}$ implies that Alice and Bob will choose $\opt_1,\opt_2$ respectively with a probability close to $1$ in following rounds.
For simplicity, we treat each event $\mathcal{E}$ as an indicator function, i.e., $\mathcal{E}$ happens if and only if $\mathcal{E}=1$. In order to prove that these two bad events cannot go on forever, we want to show that when $\mathcal{E}_t^{1,2}$ happens, $R_{1,t}-R_{2,t}$ will tend to decrease at a rapid rate. Therefore, Alice will eventually deviate from $\opt_1$ to choose $\opt_2$ with a relatively high probability.

By \cref{asm.fullsupf}, exchangeability and \cref{lemma.boundeduninformative,lemma.sumofpayoffs}, given $\delta>0$ we choose an integer $c_1\geq2$ such that, for each $i\in\{1,2\}$, when $R_{i,t}>c_1$, Alice chooses $\opt_i$ in round $t+1$ with conditional probability larger than $1-\delta$; when $S_{i,t}>c_1$, the same holds for Bob.
Let $\gamma_1 = P_{X,Y}(1,1)+P_{X,Y}(0,0)-P_{X,Y}(1,0)-P_{X,Y}(0,1)$ $\gamma_2 = (P_{X,Y}(1,1)-P_{X,Y}(0,0))^2-(P_{X,Y}(1,0)-P_{X,Y}(0,1))^2$.
\begin{lemma}\label{lemma.delta}
Given the game defined in \cref{thm.converge}, there exists a $\delta>0$ and corresponding $c_1$ such that for any round $t$, $
    \mathbb{E}[r_{1,t+1}-r_{2,t+1}\mid\mathcal{H}_{t-1}]\geq \frac{\gamma_1-\gamma_2}{2}>0\quad\text{on }\{S_{1,t-1}>c_1+1\}.
$
\end{lemma}

Given such $\delta$ and $c_1$ in \cref{lemma.delta}, we set $c_0=c_1+\lceil\frac{1000}{\gamma_1-\gamma_2}\rceil+1$. Because in each round $R_{1,t},R_{2,t}$ and $S_{1,t},S_{2,t}$ vary by at most $1$, if $\mathcal{E}_{t}^{1,2}$ happens, then in each of the next $\lceil\frac{1000}{\gamma_1-\gamma_2}\rceil+1$ rounds, conditional on its preceding history, Alice chooses $\opt_1$ and Bob chooses $\opt_2$ each with probability larger than $1-\delta$. Similar argument holds for $\mathcal{E}_{t}^{2,1}$.  We use the above observation to show \cref{lemma.expectr12}.

\begin{lemma}\label{lemma.expectr12}
Given the game defined in \cref{thm.converge}, for all $t$ and history $\mathcal{H}_t\in \mathcal{E}_t^{1,2}$, we have
$
    \mathbb{E}\left[\sum_{j=1}^{\lceil\frac{1000}{\gamma_1-\gamma_2}\rceil+1}(r_{1,t+j}-r_{2,t+j})\middle|\mathcal{H}_t\right]\leq -100.
$
\end{lemma}

This lemma formalizes the blue arrows in Fig.~\ref{Fig.schemfig}.  With this lemma, we get the main result of step two. 

\begin{lemma}\label{lemma.alwaysnicesituation}
Given the game defined in \cref{thm.converge}, $
    \Pr\left\{\limsup_{t\rightarrow\infty}\overline{\mathcal{E}_t^{1,2}\vee\mathcal{E}_t^{2,1}}=1\right\}=1.
$
\end{lemma}

If we treat $R_{1,t}-R_{2,t}$ as money of Alice, Lemma~\ref{lemma.alwaysnicesituation} is similar to the gambler's ruin problem \cite{epstein2012theory}. More specifically, $R_{1,t}-R_{2,t}$ has a negative expected growth each $\lceil\frac{1000}{\gamma_1-\gamma_2}\rceil+1$ rounds by Lemma~\ref{lemma.expectr12}, so $R_{1,t}-R_{2,t}$ will always become small enough to escape $\mathcal{E}^{1,2}$. 

\textbf{Step 3: From Intermediate States to Good Events}

We define a series of "good events" at first.  For all $u\in \N^+$ and $t\ge 1$,
$\mathcal{E}^{1,1}_t(u):=\{R_{1,t}\geq u, \text{ and }S_{1,t}\geq u\}\text{ and }\mathcal{E}^{2,2}_t(u):=\{R_{2,t}\geq u\text{ and }S_{2,t}\geq u\}$.

To the end, we want to show that $\left(\vee_{t\in\mathbb{N}^+}\mathcal{E}^{1,1}_t(u)\right)\vee\left(\vee_{t\in\mathbb{N}^+}\mathcal{E}^{2,2}_t(u)\right)$ happens with probability $1$ for any $u\in\mathbb{N}^+$. Formally, we claim Lemma~\ref{lemma.inflargegap}.

\begin{lemma}\label{lemma.inflargegap}
Given the game defined in \cref{thm.converge}, for all $u$ there exists $\lambda_u>0$ so that for any $T$ with history $\mathcal{H}_T\in \overline{\mathcal{E}_T^{1,2}\vee\mathcal{E}_T^{2,1}}$, we have
$\Pr\left\{\left(\vee_{i=T}^{T+4(u+c_0)+100}\mathcal{E}^{1,1}_i(u)\right)\vee\left(\vee_{i=T}^{T+4(u+c_0)+100}\mathcal{E}^{2,2}_i(u)\right)=1\middle|\mathcal{H}_T\right\}\geq \lambda_u$.
\end{lemma}

This lemma generally says that when the agents are in an intermediate state that $\overline{\mathcal{E}_T^{1,2}\vee\mathcal{E}_T^{2,1}}=1$, they have a constant probability $\lambda_u$ to get into a good event in the next $4(u+c_0)+100$ rounds. In order to prove this lemma, we define a nice event $\mathcal{P}_T$ such that  $\mathcal{P}_T$ happens with probability larger than $\lambda_u$ and $\mathcal{P}_T$ implies the good events happen in no more than $4(u+c_0)+100$ rounds. Formally, event $\mathcal{P}_T$ is defined as the following: First, for $j = T,\dots, T_1-1$ until some round $T_1\geq T$ such that $S_{1,T_1}\geq 0$, $x_{j+1}=1-\hat{y}_j,y_{j+1}=1-\hat{x}_j$, Alice uses strategy $\opt_1$, Bob uses strategy $\opt_2$.  Then, Alice and Bob uses strategy $\opt_1$ for $4u+50$ rounds and signals are generated as $x_{j+1}=y_{j+1}=1-x_j$ for $j\geq T_1$.
We can use Lemma~\ref{lemma.nonoptimality34} to prove that $ \lambda_u=(\min_{i,j\in\{0,1\}}\{P_{X,Y}(i,j)\}\times0.25^2)^{100+4(u+c_0)}$ is a feasible lower bound.

\textbf{Step 4: Good Events Lead to Truthful Convergence}

Similar to step 2, we have \cref{lemma.expectr12increase} to show that $\mathcal{E}_t^{1,1}(u)$ can lead to increasing of $R_{1,t}-R_{2,t}$ at first. The idea is completely similar to Lemma~\ref{lemma.expectr12} and it formalize the red arrows in Fig.~\ref{Fig.schemfig}. 

\begin{lemma}\label{lemma.expectr12increase}
Given the game defined in \cref{thm.converge}, for all $u>2c_0+1$ and $t$ if $\mathcal{H}_t\in \mathcal{E}_t^{1,1}\left(\lfloor\frac{u}{2}\rfloor\right)$, we have $
    \mathbb{E}\left[\sum_{j=1}^{\lceil\frac{1000}{\gamma_1-\gamma_2}\rceil+1}\left(r_{1,t+j}-r_{2,t+j}\right)\middle|\mathcal{H}_t\right]\geq 100
$.
\end{lemma}

Using Lemma~\ref{lemma.expectr12increase}, we are able to prove that for any $\varepsilon$, there exists a $u$ such that when $\mathcal{E}_t^{1,1}(u)$ or $\mathcal{E}_t^{2,2}(u)$ happens, Alice and Bob will tend to choose $(\opt_1,\opt_1)$ or $(\opt_2,\opt_2)$ with increasingly higher probability. Formally, we propose Lemma~\ref{lemma.convergecondition}.

\begin{lemma}\label{lemma.convergecondition}
Given the game defined in \cref{thm.converge}, for all $\varepsilon>0$ there exists $u\in\mathbb{N}^+$ such that given a history $\mathcal{H}_{T}\in \mathcal{E}_T^{1,1}(u)\vee\mathcal{E}_T^{2,2}(u)$, we  have $
    \Pr\left\{\forall i\in\mathbb{N}, \mathcal{E}_{T+\left(\lceil\frac{1000}{\gamma_1-\gamma_2}\rceil+1\right)i}^{1,1}\left(\lfloor\frac{u}{2}\rfloor+i\right)\vee\mathcal{E}_{T+\left(\lceil\frac{1000}{\gamma_1-\gamma_2}\rceil+1\right)i}^{2,2}\left(\lfloor\frac{u}{2}\rfloor+i\right)=1\middle|\mathcal{H}_T\right\}\geq 1-\varepsilon$.
\end{lemma}

We design a stopped sub-martingale $\{D_i\}_{i\in\mathbb{N}}$ based on $R_{1,T+i\left(\lceil\frac{1000}{\gamma_1-\gamma_2}\rceil+1\right)}-R_{2,T+i\left(\lceil\frac{1000}{\gamma_1-\gamma_2}\rceil+1\right)}$, and use Azuma-Hoeffding inequality to prove \cref{lemma.convergecondition}.

\section{Simulations}
\label{sec.simu}
We simulate the CA mechanism with various learning algorithms: the Hedge algorithms, follow the perturbed leader, follow the leader, and $\epsilon$-greedy, and repeat the process 400 times with 800 rounds on each algorithm each time.  In the linked plotting code, the \emph{converge proportion} in round $t$ counts each agent separately: a repetition contributes $1/(2\times400)$ for that agent if both final options agree and that agent uses her final option throughout rounds $t$ through 800. The common final option may be uninformative. This statistic is nondecreasing in $t$ by construction and does not establish truthful convergence.

The originally listed weights $(0.4,0.2,0.2,0.4)$, in the order $(00,01,10,11)$, sum to $1.2$. The linked cumulative-interval sampler on $[0,1)$ instead generates $(0.4,0.2,0.2,0.2)$, which does not satisfy \cref{asm.poscorr}. The original figures are retained without new simulations and provide finite-horizon illustrations rather than a verification of \cref{thm.converge} under its stated signal assumptions.
Moreover, Alice and Bob are using the same learning algorithms in our simulations that are listed below:
First, Follow the Leader algorithm (FTL) chooses $\opt_i$ with probability proportional to $\mathbb{I}[R_i=\max\{R_1,R_2,R_3,R_4\}]$.  Follow the perturbed leader (FPL*, where * can be 2, 4 or 8) adds a uniform random noise between $0$ and $*$ and choose strategy $\opt_i$ with probability $\Pr\left\{R_i+p_i=\max_{j\in[4]}\{R_j+p_j\}\middle|p_j\overset{\mathrm{iid}}{\sim}\mathcal{U}[0,*]\right\}$.  We consider FPL2, FPL4, and FPL8. Hedge algorithm 1 choose $i$ with probability proportional to $3^{R_i/2}$ that is an implementation by choosing $\epsilon=0.5$ for the multiplicative weights algorithm of \citet{arora2012multiplicative}. Hedge algorithm 2 chooses $i$ with probability proportional to $e^{R_i}$ that is an implementation by choosing $\beta=1$ of exponentially weighted averaged forecaster introduced by~\citet{freund1997decision}. Finally, $\epsilon$-greedy algorithm uses time varying $\epsilon=1/(n+1)^2$ after $n$ completed rounds. Note that $\epsilon$-greedy is not in $\mathcal{A}$.

In Figure~\ref{Fig.convrate}, we compare the finite-horizon converge proportions of the algorithms.   Moreover, although FTL is generally not no-regret, CA mechanism still works well with it.  Additionally, we observe that less exploration can give larger finite-horizon converge proportions (e.g. FPL4 vs FPL8), but need not further increase them (e.g. FTL vs FPL2).  Finally, we also include the $\epsilon$-greedy algorithm with time-decreasing $\epsilon$, which is not a reward-based online learning algorithm. Truthful convergence beyond reward-based online learning algorithms remains an interesting question.

\begin{figure}[htbp] 
\centering
\includegraphics[width=0.8\textwidth]{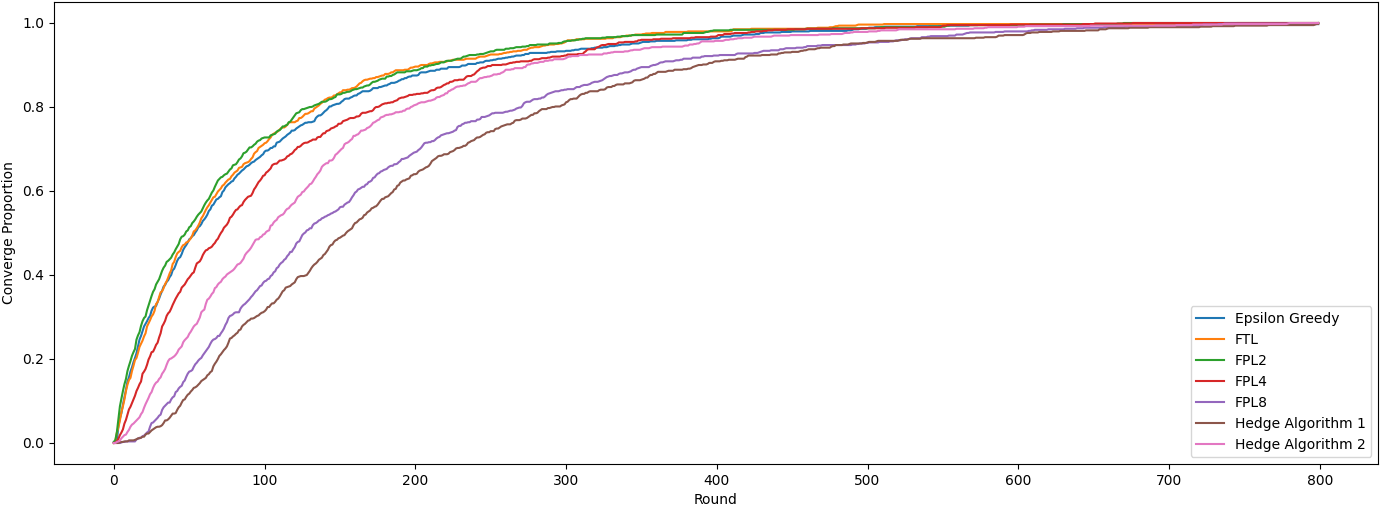}
\caption{Converge Proportions of Learning Algorithms in CA.}
\label{Fig.convrate}
\end{figure}

\section{Conclusions}

In this paper, we study sequential peer prediction with learning agents and prove that the notion of no-regret alone is not sufficient for truthful convergence. We then define a family of reward-based learning algorithms and show that the CA mechanism is able to achieve truthful convergence when agents use algorithms in this family satisfying the additional condition on full exploitation in \cref{thm.converge}. Finally, we give a discussion on the converge rates of different learning agents based on simulations.

This is a theoretical study on peer prediction with learning agents. There are many open problems and future directions to extend this work. We believe similar proof techniques can be used to extend our results to settings where agents' private signals are generated by a Markov chain with some assumptions on the transition matrix. %
Moreover, this work is only restricted to binary signals and it is still an open problem whether there exists a mechanism for non-binary settings that can promise truthful convergence. For the learning agents, one could consider a more general family of learning algorithms such as when $f$ is time-varying. 

\begin{ack}
The authors would like to thank the anonymous reviewers for their valuable comments and constructive feedback. We thank Sergey Alexeev for identifying an issue with the convergence claim in Theorem~\ref{thm.converge} of the NeurIPS 2022 version, which this revision addresses by adding a summability condition. This work is partially supported by the National Science Foundation under Grant No. IIS 2007887 and by the National Science Foundation and Amazon under Grant No. FAI 2147187.
\end{ack}

\bibliography{reference}

\OnlyInFull{
\newpage
\section*{Checklist}

\begin{enumerate}

\item For all authors...
\begin{enumerate}
  \item Do the main claims made in the abstract and introduction accurately reflect the paper's contributions and scope?
    \answerYes{}
  \item Did you describe the limitations of your work?
    \answerYes{}
  \item Did you discuss any potential negative societal impacts of your work? 
    \answerNA{}Our work focuses on pure theoretical problems, thus it does not have any negative social impact.
  \item Have you read the ethics review guidelines and ensured that your paper conforms to them?
    \answerYes{}
\end{enumerate}

\item If you are including theoretical results...
\begin{enumerate}
  \item Did you state the full set of assumptions of all theoretical results?
    \answerYes{}
        \item Did you include complete proofs of all theoretical results?
    \answerYes{}
\end{enumerate}

\item If you ran experiments...
\begin{enumerate}
  \item Did you include the code, data, and instructions needed to reproduce the main experimental results (either in the supplemental material or as a URL)?
    \answerYes{}Our code is included in supplementary materials. 
  \item Did you specify all the training details (e.g., data splits, hyperparameters, how they were chosen)? 
    \answerYes{}Training details are specified in \cref{sec.simu}.
        \item Did you report error bars (e.g., with respect to the random seed after running experiments multiple times)?
    \answerYes{} We give error bars in \cref{app.errorbar} instead of in \cref{sec.simu} for clearness of \cref{Fig.convrate}.
        \item Did you include the total amount of compute and the type of resources used (e.g., type of GPUs, internal cluster, or cloud provider)?
    \answerYes{} We illustrate our hardware details in \cref{app.errorbar}. 
\end{enumerate}

\item If you are using existing assets (e.g., code, data, models) or curating/releasing new assets...
\begin{enumerate}
  \item If your work uses existing assets, did you cite the creators?
    \answerNA{}
  \item Did you mention the license of the assets?
    \answerNA{}
  \item Did you include any new assets either in the supplemental material or as a URL?
    \answerNA{}
  \item Did you discuss whether and how consent was obtained from people whose data you're using/curating?
    \answerNA{}
  \item Did you discuss whether the data you are using/curating contains personally identifiable information or offensive content?
    \answerNA{}
\end{enumerate}

\item If you used crowdsourcing or conducted research with human subjects...
\begin{enumerate}
  \item Did you include the full text of instructions given to participants and screenshots, if applicable?
    \answerNA{}
  \item Did you describe any potential participant risks, with links to Institutional Review Board (IRB) approvals, if applicable?
    \answerNA{}
  \item Did you include the estimated hourly wage paid to participants and the total amount spent on participant compensation?
    \answerNA{}
\end{enumerate}

\end{enumerate}
}

\clearpage
\appendix
\section*{Appendix}

\section{Basic Math}

\subsection{Martingale and Concentration}

In this section we will define martingales and some of its properties.
\begin{definition}[Martingale]
Let $\mathcal{F}=\{\mathcal{F}_t\}_{t\in\mathbb{N}}$ be a filtration, which is an increasing sequence of $\sigma$-field.
A \emph{martingale} with respect to $\mathcal{F}$ is a sequence $D_0, D_1, D_2,\cdots$ adapted to $\mathcal{F}$ ($D_t\text{ is }\mathcal{F}_t\text{-measurable}$ for all $t$) that satisfies for any time $t$,\begin{align*}
    &\mathbb{E}\left[|D_t|\right]<\infty,\\
    &\mathbb{E}\left[D_{t+1}|\mathcal{F}_t\right]=D_{t}.
\end{align*}
\end{definition}

There are two extensions of a martingale that replace the equality of conditional expectation by upper and lower bounds.
\begin{definition}[Sub-martingale and super-martingale] \yiling{}
Let $\mathcal{F}=\{\mathcal{F}_t\}_{t\in\mathbb{N}}$ be a filtration, which is an increasing sequence of $\sigma$-field.
A \emph{sub-martingale} with respect to $\mathcal{F}$ \yiling{}\shi{}is a sequence $D_0, D_1, D_2,\cdots$ adapted to $\mathcal{F}$ that satisfies for any time $t$,\begin{align*}
    &\mathbb{E}\left[|D_t|\right]<\infty,\\
    &\mathbb{E}\left[D_{t+1}|\mathcal{F}_t\right]\geq D_{t}.
\end{align*}
A \emph{super-martingale} with respect to $\mathcal{F}$ is a sequence $D_0, D_1, D_2,\cdots$ adapted to $\mathcal{F}$ that satisfies for any time $t$,\begin{align*}
    &\mathbb{E}\left[|D_t|\right]<\infty,\\
    &\mathbb{E}\left[D_{t+1}|\mathcal{F}_t\right]\leq D_{t}.
\end{align*}
\end{definition}

For a sub-martingale (or super-martingale), we have the Azuma--Hoeffding inequality \cite{alon2016probabilistic}, which is a concentration result for the values of martingales.
\begin{theorem}[name = Azuma--Hoeffding inequality, label = thm:asuma]
Suppose $\{D_t\}_{t\in\mathbb{N}}$ is a martingale (or super-martingale) and
$|D_t-D_{t-1}|\leq c_t$ almost surely, where $c_t\geq0$ are deterministic constants. Then for any $T\in\mathbb{N}^+$ and $\epsilon\in\mathbb{R}^+$, we have\begin{align*}
    \Pr\left\{D_T-D_0\geq\epsilon\right\}\le\exp\left(\frac{-\epsilon^2}{2\sum_{t=1}^Tc_t^2}\right).
\end{align*}
And symmetrically, if $\{D_t\}_{t\in\mathbb{N}}$ is a sub-martingale, we have\begin{align*}
    \Pr\left\{D_T-D_0\leq-\epsilon\right\}\le\exp\left(\frac{-\epsilon^2}{2\sum_{t=1}^Tc_t^2}\right).
\end{align*}
\end{theorem}

Now we define stopping time. This is intuitively a condition such that the "decision" whether to stop in the $t^{th}$ round should be based on information of the first $t$ rounds, instead of any future information.
\begin{definition}[Stopping time]
$\tau$ is called a \emph{stopping time} for a filtration $\mathcal{F}$ if and only if $\{\tau=t\}\in\mathcal{F}_t,\forall t$.
\end{definition}

\subsection{Limit Inferior and Limit Superior}

Limit inferior and limit superior are defined on sequences, representing limit bounds of a sequence. To meet our needs, our definition focus on discrete metric.

\begin{definition}
Let $\{\mathcal{E}_t\}_{t\in\mathbb{N}^+}$ be a sequence of events. Limit inferior and limit superior of this sequence are\begin{align*}
    \liminf_{t\rightarrow+\infty}\mathcal{E}_t=\vee_{t=1}^{+\infty}\left(\wedge_{i=t}^{+\infty}\mathcal{E}_i\right),
\end{align*}
and\begin{align*}
    \limsup_{t\rightarrow+\infty}\mathcal{E}_t=\wedge_{t=1}^{+\infty}\left(\vee_{i=t}^{+\infty}\mathcal{E}_i\right).
\end{align*}
\end{definition}

\subsection{Borel-Cantelli Lemma}

In this section, we formally introduce Borel-Cantelli lemma. Its proof can be found in \cite{feller2008introduction}.
\begin{theorem}[Borel-Cantelli lemma]\label{thm.borel}
Let $\mathcal{E}_1,\mathcal{E}_2,\cdots$ be a sequence of events in some probability space. Borel-Cantelli lemma states that if $\sum_{t=1}^{+\infty}\Pr\{\mathcal{E}_t\}<\infty$, then the probability that infinitely many of $\mathcal{E}_t$'s occur is $0$, or more strictly,\begin{align*}
    \Pr\left\{\limsup_{t\rightarrow+\infty}\mathcal{E}_t\right\}=0.
\end{align*}
\end{theorem}

The conditional Borel--Cantelli implication used below is as follows: if $A_t\in\mathcal{H}_t$ and $p_t=\Pr(A_t\mid\mathcal{H}_{t-1})$, then $\sum_tp_t<\infty$ almost surely implies that only finitely many $A_t$ occur almost surely. Indeed, for every integer $m\geq1$, predictability and Tonelli's theorem give
\[
    \mathbb{E}\sum_t\mathbb{I}[A_t]\mathbb{I}\!\left[\sum_{s<t}p_s\leq m\right]
    =\mathbb{E}\sum_tp_t\mathbb{I}\!\left[\sum_{s<t}p_s\leq m\right]\leq m+1.
\]
The selected count is finite almost surely; when $\sum_tp_t<\infty$, some integer $m$ includes every occurrence, proving the implication.

\section{Truthfulness of CA mechanism for Bayesian Agents}
\label{app.bayesiantruthful}
In this section, we assume that both agents are using consistent strategies as previous works assume \cite{dasgupta2013crowdsourced,shnayder2016informed,kong2016framework,schoenebeck2020learning,kong2020dominantly,DBLP:journals/corr/abs-2106-03176}. The formal definition of consistent strategy is given in \cref{def.consist}.
\begin{definition}[name = Consistent strategy, label = def.consist]
In a repeated game, a strategy profile $\sigma$ is a \emph{consistent strategy} if and only if for each agent $i$, she adopts $\sigma_i$ identically over each round of the game.
\end{definition}
In a sequential peer prediction game, for agent $i$, we set the average payoffs of agent $i$ in $T$ rounds as her utility.

Then we introduce Bayesian Nash equilibrium. A Bayesian Nash equilibrium (BNE) is strategies of agents that have the maximal expected payoff for each player given their beliefs on environments and others' strategies \cite{kajii1997robustness}. 
\begin{definition}[Bayesian Nash equilibiurm]
A strategy profile $\sigma$ is a \emph{Bayesian Nash equilibrium} if and only if for every agent $i$, the expected payoff of using $\sigma_i$ for agent $i$ is maximal keeping other agents' strategies unchanged. Moreover, A strategy profile $\sigma$ is a \emph{strict Bayesian Nash equilibrium} if and only if for every agent $i$, the expected payoff of using $\sigma_i$ for agent $i$ is strictly larger than any other strategies keeping other agents' strategies unchanged. 
\end{definition}

Moreover, We say a peer prediction mechanism is \emph{strongly truthful} if agents in truthtelling equilibrium get strictly higher payment than any other non-permutation equilibrium \cite{shnayder2016informed}. Here, a permutation equilibrium is the strategy profile that agents report a permutation of the signal. Formally, we have \cref{def.strongtruth}.
\begin{definition}[Strongly truthful]\label{def.strongtruth}
In a peer prediction game, if agents are using consistent strategies, a mechanism is \emph{strongly truthful} if and only if truthtelling is a BNE and also guarantees larger agent welfare than any non-permutation equilibrium. Here, welfare is defined by each agent's expected payoff so that is to say, the expected payoff of each agent using truthtelling strategy profile is strictly higher than the expected payoff using non-permutation equilibrium.
\end{definition}

Though \citet{dasgupta2013crowdsourced,shnayder2016informed} have proved that CA mechanism is strongly truthful in non-sequential settings, our settings are slightly different from theirs and we rewrite the proof for binary sequential peer prediction settings. More specifically, our CA mechanism uses the last round agreement term instead of average agreement term in the payoffs and we are focusing on average payoffs instead of total payoffs. Before the complete proof, we have the following lemma.
\begin{lemma}\label{thm.strictne}
For binary sequential signal peer prediction games under \cref{asm.apriori} and \cref{asm.poscorr}, if both agents are using consistent strategies, CA mechanism renders truthtelling strategy profile a strict Bayesian Nash equilibrium.
\end{lemma}

\begin{proof}
We know that Alice and Bob are Bayesian agents using consistent strategies. Then we can suppose when Alice gets signal $1$, she reports $1$ with probability $p_1$; when she gets signal $0$, she reports $0$ with probability $p_0$. Similarly, Bob is using a fixed strategy that when he gets signal $1$, he reports $1$ with probability $q_1$; when he gets signal $0$, he reports $0$ with probability $q_0$. Therefore, we can compute the first term in payoff of Alice and Bob (see \cref{eq:ca}), $\mathbb{E}[\mathbb{I}[\hat{x}_t=\hat{y}_t]]$, by \begin{align*}
    \mathbb{E}[\mathbb{I}[\hat{x}_t=\hat{y}_t]]&=P_{X,Y}(1,1)(p_1q_1+(1-p_1)(1-q_1))+P_{X,Y}(0,0)(p_0q_0+(1-p_0)(1-q_0))\\
    &+P_{X,Y}(1,0)(p_1(1-q_0)+(1-p_1)q_0)+P_{X,Y}(0,1)(q_1(1-p_0)+(1-q_1)p_0).
\end{align*} 
Also, we can compute the second term\begin{small}\begin{align*}
    \mathbb{E}[\mathbb{I}[\hat{x}_t=\hat{y}_{t-1}]]&=(P_{X,Y}(1,1)+P_{X,Y}(1,0))(P_{X,Y}(0,1)+P_{X,Y}(1,1))(p_1q_1+(1-p_1)(1-q_1))\\
    &+(P_{X,Y}(0,0)+P_{X,Y}(0,1))(P_{X,Y}(1,0)+P_{X,Y}(0,0))(p_0q_0+(1-p_0)(1-q_0))\\
    &+(P_{X,Y}(1,1)+P_{X,Y}(1,0))(P_{X,Y}(0,0)+P_{X,Y}(1,0))(p_1(1-q_0)+(1-p_1)q_0)\\
    &+(P_{X,Y}(0,1)+P_{X,Y}(0,0))(P_{X,Y}(0,1)+P_{X,Y}(1,1))((1-p_0)q_1+p_0(1-q_1)).
\end{align*}\end{small}
Similarly, $\mathbb{E}[\mathbb{I}[\hat{y}_t=\hat{x}_{t-1}]]$ also equals to this expression.

To prove BNE, we only need to prove that when $q_0=q_1=1$, we have\begin{align}
    \argmax_{(p_0,p_1)}\left\{\mathbb{E}[\mathbb{I}[\hat{x}_t=\hat{y}_t]]-\mathbb{E}[\mathbb{I}[\hat{x}_t=\hat{y}_{t-1}]]\right\}=(1,1).\label{eq.strictne}
\end{align}
The other side when $p_0=p_1=1$ is symmetrical.

Actually, we have\begin{align*}
    &\frac{\partial (\mathbb{E}[\mathbb{I}[\hat{x}_t=\hat{y}_t]|q_0=q_1=1]-\mathbb{E}[\mathbb{I}[\hat{x}_t=\hat{y}_{t-1}]|q_0=q_1=1])}{\partial p_0}\\
    &=P_{X,Y}(0,0)-P_{X,Y}(0,1)\\
    &+(P_{X,Y}(0,0)+P_{X,Y}(0,1))(P_{X,Y}(0,1)+P_{X,Y}(1,1)-P_{X,Y}(0,0)-P_{X,Y}(1,0))\\
    &=2P_{X,Y}(0,0)-2(P_{X,Y}(0,0)+P_{X,Y}(0,1))(P_{X,Y}(0,0)+P_{X,Y}(1,0))\\
    &=2(P_{X,Y}(0,0)P_{X,Y}(1,1)-P_{X,Y}(0,1)P_{X,Y}(1,0))\\
    &>0,
\end{align*}
and\begin{align*}
    &\frac{\partial (\mathbb{E}[\mathbb{I}[\hat{x}_t=\hat{y}_t]|q_0=q_1=1]-\mathbb{E}[\mathbb{I}[\hat{x}_t=\hat{y}_{t-1}]|q_0=q_1=1])}{\partial p_1}\\
    &=P_{X,Y}(1,1)-P_{X,Y}(1,0)\\
    &+(P_{X,Y}(1,1)+P_{X,Y}(1,0))(P_{X,Y}(0,0)+P_{X,Y}(1,0)-P_{X,Y}(1,1)-P_{X,Y}(0,1))\\
    &=2P_{X,Y}(1,1)-2(P_{X,Y}(1,1)+P_{X,Y}(1,0))(P_{X,Y}(1,1)+P_{X,Y}(0,1))\\
    &=2(P_{X,Y}(1,1)P_{X,Y}(0,0)-P_{X,Y}(0,1)P_{X,Y}(1,0))\\
    &>0,
\end{align*}
which indicate \cref{eq.strictne}.
\end{proof}

Furthermore, we can find all the Nash equilibria for Bayesian agents under our CA mechanism in sequential peer prediction setting. We have \cref{thm.bayesianne}.
\begin{theorem}\label{thm.bayesianne}
For binary sequential signal peer prediction games under \cref{asm.apriori} and \cref{asm.poscorr}, if CA mechanism is used and agents adopt consistent strategies, there are three types of Nash equilibria for agents, which are\begin{enumerate}
    \item truth-telling $\opt_1$,
    \item flip the signal $\opt_2$,
    \item report regardless of private signals (uninformative reports). 
\end{enumerate} 
\end{theorem}

\begin{proof}
Using the same notations in the proof of \cref{thm.strictne}, we deduce that \begin{align*}
    &\frac{\partial (\mathbb{E}[\mathbb{I}[\hat{x}_t=\hat{y}_t]]-\mathbb{E}[\mathbb{I}[\hat{x}_t=\hat{y}_{t-1}]])}{\partial p_0}\\
    &=(P_{X,Y}(1,1)P_{X,Y}(0,0)-P_{X,Y}(1,0)P_{X,Y}(0,1))(2q_0+2q_1)\\
    &-2(P_{X,Y}(1,1)P_{X,Y}(0,0)-P_{X,Y}(1,0)P_{X,Y}(0,1)),
\end{align*}
and similarly,\begin{align*}
    &\frac{\partial (\mathbb{E}[\mathbb{I}[\hat{x}_t=\hat{y}_t]]-\mathbb{E}[\mathbb{I}[\hat{x}_t=\hat{y}_{t-1}]])}{\partial p_1}\\
    &=(P_{X,Y}(1,1)P_{X,Y}(0,0)-P_{X,Y}(1,0)P_{X,Y}(0,1))(2q_0+2q_1)\\&-(P_{X,Y}(1,1)P_{X,Y}(0,0)-P_{X,Y}(1,0)P_{X,Y}(0,1)).
\end{align*}
Therefore, we know that when $q_0+q_1>1$, the best response of Alice is $p_0=p_1=1$; when $q_0+q_1<1$, the best response of Alice is $p_0=p_1=0$. Symmetrically, we can deduce that when $p_0+p_1>1$, the best response of Alice is $q_0=q_1=1$; when $p_0+p_1<1$, the best response of Alice is $q_0=q_1=0$.

Let $p_0^*,p_1^*,q_0^*,q_1^*$ be a Nash equilibrium for this peer prediction game under CA mechanism. Then given $p_0^*,p_1^*$, $q_0^*,q_1^*$ should be one of the best responses of Bob; given $q_0^*,q_1^*$, $p_0^*,p_1^*$ should be one of the best responses of Alice.

If $p_0^*+p_1^*>1$, then the only best response of Bob is $q_0=q_1=1$, so $q^*_0=q^*_1=1$. When $q^*_0=q^*_1=1$, the only best response of Alice is $p_0=p_1=1$, so we can deduce that $p_0^*=p_1^*=q_0^*=q_1^*=1$ is the unique Nash equilibrium for this case.

If $p_0^*+p_1^*<1$, then the only best response of Bob is $q_0=q_1=0$, so $q^*_0=q^*_1=0$. When $q^*_0=q^*_1=0$, the only best response of Alice is $p_0=p_1=0$, so we can deduce that $p_0^*=p_1^*=q_0^*=q_1^*=0$ is the unique Nash equilibrium for this case.

Symmetrically, when $q_0^*+q_1^*\neq 1$, there are still only these two Nash equilibria. Moreover, it is easy to verify that $p_0^*,p_1^*,q_0^*,q_1^*$ that satisfies $p_0^*+p_1^*=q_0^*+q_1^*=1$ is Nash equilibrium. Therefore, all the uninformative reports are Nash equilibria. These are exactly what we want to prove.
\end{proof}

From the proof of \cref{thm.strictne}, we can observe that when both agents use $\opt_1$, their expected payoffs are both \begin{align*}
    &P_{X,Y}(1,1)+P_{X,Y}(0,0)-(P_{X,Y}(1,1)+P_{X,Y}(1,0))(P_{X,Y}(0,1)+P_{X,Y}(1,1))\\
    &-(P_{X,Y}(0,1)+P_{X,Y}(0,0))(P_{X,Y}(1,0)+P_{X,Y}(0,0))\\
    &=2P_{X,Y}(1,1)P_{X,Y}(0,0)-2P_{X,Y}(1,0)P_{X,Y}(0,1)\\
    &>0.
\end{align*}
Also, we can observe that when both agents use uninformative reports, their expected payoffs are both $0$. According to \cref{thm.bayesianne}, we know that the only non-permutation equilibrium is uninformative reports. Therefore, we deduce that CA mechanism is strongly truthful.
\begin{theorem}
For binary sequential signal peer prediction games under \cref{asm.apriori} and \cref{asm.poscorr}, if agents adopt consistent strategies, CA mechanism is strongly truthful.
\end{theorem}

\section{Proofs and Details of Section~\ref{sec:noregret}}
\subsection{Impossibility of Truthful Convergence for No Regret Agents}\label{app:impossible}
\yiling{}\fang{}
\begin{proof}[Proof of \cref{thm:impossible}]
We construct two report sequences that do not depend on private signals with the required expected regret. The mechanism and signal marginals may be hardcoded into the construction, since only existence is claimed. The argument allows the rank-$k$ payment functions to vary arbitrarily with time.

Partition time into blocks $\mathcal{B}_b=\{T_{b-1}+1,\ldots,T_b\}$ of length $b$, where $T_b=b(b+1)/2$. Consider only the rounds $\mathcal{I}_b=\{T_{b-1}+k,\ldots,T_b\}$; this set is empty if $b<k$. Each retained payment depends entirely on the reports in its own block. Discarding the other rounds changes any cumulative regret through $T_B$ by at most $2(k-1)B$.

In an auxiliary random construction, each agent has four options in an auxiliary online decision problem on a block. Auxiliary option $i$ generates an auxiliary signal vector with independent coordinates distributed according to that agent's signal marginal, and reports $\opt_i$ of these auxiliary signals. Both agents commit their entire report vectors at the start of the block. Their fresh block random variables are conditionally independent given all previous auxiliary randomness, and are independent of the private signal stream. The agents do not use their private signals.

For Alice, let $z^X_{i,b}$ be the expected sum of her retained block payments if she uses auxiliary option $i$, with Bob's realized committed block report vector held fixed; the expectation is only over Alice's auxiliary signal vector. Thus $z^X_{i,b}\in[-b,b]$. The vector $z_b^X$ is computable after the block from Bob's reports and the announced mechanism. Define $z_b^Y$ analogously. Each agent uses the following update function for the four auxiliary options, with separate cumulative vectors:
\[
    A_{i,0}=0,\qquad A_{i,b}=\sum_{j=1}^b z_{i,j},\qquad
    p_{i,b}=\frac{\exp(A_{i,b-1}/\lambda_b)}{\sum_{j=1}^4\exp(A_{j,b-1}/\lambda_b)},
    \qquad\lambda_b=b^{3/2}.
\]
The chosen auxiliary option has distribution $p_b$. These distributions depend only on previous blocks, so the commitment and conditional independence requirements are respected.

We need a two-sided regret estimate for this update. Put
\[
    \Phi_\lambda(a)=\lambda\log\!\left(\frac14\sum_{i=1}^4e^{a_i/\lambda}\right).
\]
Convexity and Hoeffding's lemma, applied to the four values $z_{i,b}\in[-b,b]$, give
\[
    0\leq\Phi_{\lambda_b}(A_b)-\Phi_{\lambda_b}(A_{b-1})-p_b\cdot z_b
    \leq\frac{b^2}{2\lambda_b}.
\]
Moreover, the variational identity
\[
    \Phi_\lambda(a)=\max_{p\in\triangle^3}
    \left\{p\cdot a-\lambda\sum_{i=1}^4p_i\log(4p_i)\right\}
\]
and $0\leq\sum_i p_i\log(4p_i)\leq\log4$ show that, for $b\geq2$,
\[
    -\left(\lambda_b-\lambda_{b-1}\right)\log4
    \leq\Phi_{\lambda_b}(a)-\Phi_{\lambda_{b-1}}(a)\leq0.
\]
Since $\max_i a_i-\lambda\log4\leq\Phi_\lambda(a)\leq\max_i a_i$, telescoping these inequalities yields, for every realized sequence of payoff vectors,
\begin{equation}\label{eq:expect}
    \left|\max_i A_{i,B}-\sum_{b=1}^B p_b\cdot z_b\right|
    \leq\lambda_B\log4+\frac12\sum_{b=1}^B\frac{b^2}{\lambda_b}
    =O(B^{3/2}).
\end{equation}

Let $G_b^X$ be Alice's realized retained block payoff, and let $\mathcal{G}_{b-1}$ contain all auxiliary randomness through block $b-1$. Conditional on $\mathcal{G}_{b-1}$ and Bob's committed current block vector, Alice's current draw is still distributed as $p_b^X$ followed by her independent auxiliary signal vector. Consequently
\[
    \mathbb{E}[G_b^X-p_b^X\cdot z_b^X\mid\mathcal{G}_{b-1}]=0,
    \qquad |G_b^X-p_b^X\cdot z_b^X|\leq2b.
\]
The same holds for Bob. Azuma--Hoeffding bounds the probability that either cumulative difference exceeds $\varepsilon T_B$ in absolute value by
\[
    4\exp\!\left(-\frac{\varepsilon^2T_B^2}{8\sum_{b=1}^B b^2}\right),
\]
which is summable in $B$ for every fixed $\varepsilon>0$. Applying Borel--Cantelli for positive rational $\varepsilon$ and using \cref{eq:expect}, we obtain, simultaneously for both agents, almost surely,
\begin{equation}\label{eq:derandom}
    \max_i\sum_{b=1}^B z^X_{i,b}-\sum_{b=1}^B G_b^X=o(T_B),
    \qquad
    \max_i\sum_{b=1}^B z^Y_{i,b}-\sum_{b=1}^B G_b^Y=o(T_B).
\end{equation}
Fix one complete auxiliary realization on this probability-one event, giving deterministic report sequences $(\hat x_t^*,\hat y_t^*)$. The auxiliary random construction is only an existence argument; the final algorithms simply output these sequences.

We now evaluate regret using the private signals, not the fixed auxiliary signals. Let $R^X_{i,T}$ be Alice's accumulated payoff computed from private signals for the fixed reports of Bob, and let $C_T^X$ be her deterministic cumulative actual payoff. For retained rounds in complete blocks, the definition of $z^X_{i,b}$ gives $\mathbb{E}R^X_{i,T_B}=\sum_{b\leq B}z^X_{i,b}$ up to at most $(k-1)B$ discarded terms. Similarly, $C_{T_B}^X=\sum_{b\leq B}G_b^X$ up to these terms. By \cref{eq:derandom},
\[
    \max_i\mathbb{E}R^X_{i,T_B}-C_{T_B}^X=o(T_B).
\]
Changing one private signal $X_s$ changes each $R^X_{i,T}$ by at most $2k$, because only the next $k$ rank-$k$ payments can be affected. The bounded-differences moment-generating-function bound and a maximum over four coordinates therefore give
\begin{equation}\label{eq:concen}
    0\leq\mathbb{E}\!\left[\max_iR^X_{i,T}\right]-\max_i\mathbb{E}R^X_{i,T}
    \leq k\sqrt{2T\log4}.
\end{equation}
Thus $\mathbb{E}Reg^X(T_B)=o(T_B)$, and likewise for Bob. A time $T_B\leq T<T_{B+1}$ adds at most $B+1$ rounds, changing each regret by at most $2(B+1)=o(T_B)$, so the same conclusion holds for every $T\to\infty$.

Finally, implement the fixed reports by $\opt_t^X=\opt_{4-\hat x_t^*}$ and $\opt_t^Y=\opt_{4-\hat y_t^*}$. These algorithms use only $\opt_3,\opt_4$, independently of private signals, yet both have $o(T)$ expected regret. Their realized pure strategies cannot converge to $\opt_1$ or $\opt_2$, proving the theorem.

The randomized block procedure itself is also a valid implementation of the required algorithms, without hardcoding an infinite sequence. Indeed, conditioning on all auxiliary randomness leaves payoffs computed from private signals with the same uniform maximum correction in \cref{eq:concen}. Taking expectations in the block martingale identity and applying the two-sided potential bound yields
\[
    |\mathbb{E}Reg^a(T_B)|
    \leq O(B^{3/2})+O(kB)+O(k\sqrt{T_B})=o(T_B),\qquad a=X,Y.
\]
The incomplete-block bound again covers every horizon. At each round this procedure reports a previously chosen auxiliary bit, hence its realized pure strategy is always $\opt_3$ or $\opt_4$, independently of the private signal. Thus the theorem holds for the explicit randomized construction as well as for the deterministic report sequences established above.
\end{proof}

\subsection{Truthful Convergence Implies No regret}\label{app:converge2noregret}

\begin{proof}[Proof of \cref{thm:converge2noregret}]
Given reports $\hat{\vx}$ and signals $\vx$, let $\hat{\vx}_{t-k:t} = (\hat{x}_{t-k+1}, \dots, \hat{x}_{t})$ be a slice of reports $\hat{\vx}$, and $\opt_i({\vx}_{t-k:t}) = (\opt_i(x_{t-k+1}),\dots, \opt_i(x_{t}))$ be a slice reports under consistent strategy $\opt_i$.  Given $T$, we define four functions on the signals 
$$F^i_T(\vx, \vy):= \sum_{t = k}^T\left[\bar{M}^X(\opt_i(\vx_{t-k:t}), {\vy}_{t-k:t})-\bar{M}^X({\vx}_{t-k:t}, {\vy}_{t-k:t})\right]$$
for $i = 1,2,3,4$. The finitely many initialization payments do not affect the asymptotic bounds below. We want to show the value of $F^i_T$ is small.  Specifically, we bound the probability of the following good event
$$\mathcal{G}:=\{\max_i F^i_T(\vx, \vy)\le T^{2/3}\}.$$
First because $\bar{M}$ is strongly truthful, $\E_{\sigma^X}[\bar{M}^X(\hat{\vx}, \vy)]\le\E[\bar{M}^X({\vx}, \vy)]$ and $\E_{\sigma^Y}[\bar{M}^Y({\vx}, \hat{\vy})]\le \E[\bar{M}^Y({\vx}, \vy)]$, the expectation is non-positive
$$\E[F^i_T(\vx, \vy)]= \sum_{t = k}^T\E_{\sigma^X = \opt_i}\left[ \bar{M}^X(\hat{\vx}_{t-k:t}, {\vy}_{t-k:t})-\bar{M}^X({\vx}_{t-k:t}, {\vy}_{t-k:t})\right]\le 0$$ for all $i$.
Second, the signal pairs are independent across rounds, and changing one pair affects at most $k$ payment gains, each in $[-2,2]$. Thus it changes $F^i_T$ by at most $4k$. By the method of bounded differences, we have for all $i$, $\Pr\left[F^i_T(\vx, \vy)\ge T^{2/3}\right]\le \Pr\left[F^i_T(\vx, \vy)\ge T^{2/3}+\E F^i_T\right]\le \exp\left(-\frac{2}{16k^2}T^{1/3}\right)$.  By union bound, we have
\begin{equation}\label{eq:converge2noregret1}
    \Pr\left[\mathcal{G}\right]\ge 1-4\exp\left(-\frac{2}{16k^2}T^{1/3}\right).
\end{equation}

On the other hand, the truthful convergence consists of two disjoint events: both converging to truth telling $\opt_1$, and both converging to the flipping strategy $\opt_2$.  First suppose the truth-telling event happens with a nonzero probability; the flipping event will be treated separately using its assumed Bayesian Nash equilibrium
$$\mathcal{E}:=\{\lim_{t\rightarrow +\infty}\opt^X_t=\lim_{t\rightarrow +\infty}\opt^Y_t=\opt_1\}.$$
Then there exists a random round $t^*$ so that $\opt^X_t=\opt^Y_t=\opt_1$ for all $t\ge t^*$ given $\mathcal{E}$. 
To bound the expected regret conditional on $\mathcal{E}$, we consider two cases: If the converge time $t^*$ is greater than $T^{2/3}$, we use the truthful convergence to show the probability is small. Otherwise if the converge time $t^*$ is smaller than $T^{2/3}$, we can ignore the first term.

Formally, Alice's expected regret is 
$$\E\left[Reg^X(T)\mid \mathcal{E}\right] =\E\left[\mathbf{1}[t^*> T^{2/3}]Reg^X(T)\mid \mathcal{E}\right]+ \E\left[\mathbf{1}[t^*\le T^{2/3}]Reg^X(T)\mid \mathcal{E}\right].$$
For the first term, $\left|\E\left[\mathbf{1}[t^*> T^{2/3}]Reg^X(T)\mid \mathcal{E}\right]\right|\le 2T\Pr[t^*> T^{2/3}\mid \mathcal{E}]$. Because $t^*$ is finite almost surely conditional on $\mathcal{E}$, $\Pr[t^*> T^{2/3}\mid\mathcal{E}]=o(1)$. No stopping-time or integrability assumption on $t^*$ is needed. Thus we have
\begin{equation}\label{eq:converge2noregret2}
    \E\left[\mathbf{1}[t^*> T^{2/3}]Reg^X(T)\mid \mathcal{E}\right] = o(T).
\end{equation}
On the other hand, when $t^*\le T^{2/3}$ happens,
\begin{align*}
    &Reg^X(T) = \max_i\sum_{t\le T}r_{i,t}-\sum_{t\le T}r_t\\
    \le& 2(t^*+k)+\max_i\sum_{t = t^*+k}^T (r_{i,t}- r_t)\tag{$r_t$ and $r_{i,t}$ are in $[-1,1]$}\\
    =& 2(t^*+k)+\max_i\sum_{t = t^*+k}^T\left[\bar{M}^X(\opt_i(\vx_{t-k:t}), {\vy}_{t-k:t})-\bar{M}^X({\vx}_{t-k:t}, {\vy}_{t-k:t})\right]\tag{$\mathcal{E}$ happens}\\
    \le& 4(t^*+k)+\max_i F^i_T(\vx,\vy)\tag{adding addition $t^*+k$ terms}\\
    \le& 4(T^{2/3}+k)+\max_i F^i_T(\vx,\vy)\tag{$t^*\le T^{2/3}$}
\end{align*}
Therefore, 
\begin{equation}\label{eq:converge2noregret3}
    \E\left[\mathbf{1}[t^*\le T^{2/3}]Reg^X(T)\mid \mathcal{E}\right]\le 4(T^{2/3}+k)+\E\left[\max_i F^i_T(\vx,\vy)\mid \mathcal{E}\right].
\end{equation}
To bound the second term, we partition the expectation by whether $\mathcal{G}$ happens or not
\begin{align*}
    &\E\left[\max_i F^i_T(\vx,\vy)\mid \mathcal{E}\right]\\
    =& \E\left[\max_i F^i_T(\vx,\vy)\mid \mathcal{E}, \mathcal{G}\right]\Pr[\mathcal{G}\mid \mathcal{E}]+\E\left[\max_i F^i_T(\vx,\vy)\mid \mathcal{E}, \neg\mathcal{G}\right]\Pr[\neg\mathcal{G}\mid \mathcal{E}]\\
    \le& T^{2/3}\Pr[\mathcal{G}\mid \mathcal{E}]+2T\Pr[\neg\mathcal{G}\mid \mathcal{E}]\tag{definition of $\mathcal{G}$}\\
    \le& T^{2/3}+2T\frac{\Pr[\neg\mathcal{G}]}{\Pr[\mathcal{E}]} = O(T^{2/3})\tag{by \cref{eq:converge2noregret1}}
\end{align*}
Therefore, with \cref{eq:converge2noregret2,eq:converge2noregret3} we obtain the upper bound $\E\left[Reg^X(T)\mid \mathcal{E}\right]\le o(T)+4(T^{2/3}+k)+O(T^{2/3}) = o(T)$.
For the matching lower bound, the comparator $\opt_1$ agrees with the actual reports after $t^*$, so $Reg^X(T)\ge-2\min\{T,t^*+k\}$ on $\mathcal{E}$. Splitting again at $T^{2/3}$ gives
\[
\E[Reg^X(T)\mid\mathcal{E}]\ge-2(T^{2/3}+k)-2T\Pr[t^*>T^{2/3}\mid\mathcal{E}]=-o(T).
\]
Hence $\E[Reg^X(T)\mid\mathcal{E}]=o(T)$. For the common-flipping branch, define the analogous $F_T^i$ with both ideal reports given by $\opt_2$. Its assumed Bayesian Nash equilibrium gives the same nonpositive mean gains, and $F_T^2=0$. Apply the same concentration bound under the original i.i.d. signal law, then condition using $\Pr[\neg\mathcal{G}\mid\mathcal{E}_{\rm flip}]\le\Pr[\neg\mathcal{G}]/\Pr[\mathcal{E}_{\rm flip}]$ when this branch has positive probability. The same argument therefore proves $o(T)$ conditional expected regret on that branch. Weighting the two branches proves $\E[Reg^X(T)]=o(T)$; Bob's proof is identical.
\end{proof}
\section{Justifications of Reward-Based Online Learning Algorithm Family \texorpdfstring{$\mathcal{A}$}{Lg}}

In this section, we have two subsections to give justifications for learning algorithm family $\mathcal{A}$. In the first part, we introduce two common used learning algorithms and show that they are both reward-based online learning algorithms. The second part gives justifications for \cref{asm.fullsupf}.

\subsection{Learning Algorithms in \texorpdfstring{$\mathcal{A}$}{Lg}}
\label{app.algina}

In this section, we do not focus on peer prediction problems but consider learning algorithms used on general online decision problems.

\subsubsection{Follow the Perturbed Leader}

We consider FPL as studied in \cite{kalai2005efficient}. Their no-regret bounds in full-information online decision problems use suitably tuned perturbations; here we use a fixed perturbation law. For simplicity, we consider FPL using on online decision problem with four options $\opt_i,i\in[4]$ in total. The algorithm let the agent choose an option uniformly among $\argmax_{\opt_i\in\{\opt_1,\opt_2,\opt_3,\opt_4\}}(R_{i,t-1}+p_{i,t})$ in the $t^{th}$ round of the game. Here, $p_{i,t}$'s are i.i.d. sampled from a particular noise distribution $\mathcal{N}$. Because of variety of the noise distribution, FPL algorithm actually contains a large family of learning algorithms, i.e., Hannan's algorithm \cite{hannan1957approximation} and Follow the Leader algorithm (FTL).

\begin{algorithm}
    \caption{FPL algorithm.}\label{alg:fpl}
    \DontPrintSemicolon
    \KwIn{Noise distribution $\mathcal{N}$.}
    \For{$t=1,2,\cdots$}{
        For each option $\opt_i,i\in[4]$, sample $p_{i,t}$ independently from noise distribution $\mathcal{N}$.\\
        Uniformly choose an option among $\argmax_{\opt_i\in\{\opt_1,\opt_2,\opt_3,\opt_4\}}(R_{i,t-1}+p_{i,t})$ in the $t^{th}$ round.
    }
\end{algorithm}

In formal, we have the following theorem.
\begin{theorem}
Algorithm~\ref{alg:fpl} is a reward-based online learning algorithm included in $\mathcal{A}$.
\end{theorem}

\begin{proof}
To prove that FPL algorithm is in our algorithm family $\mathcal{A}$, we need to design a function $f^{\text{FPL}}$ such that \begin{align*}
f^{\text{FPL}}_i(R_1,R_2,R_3,R_4)=\mathbb{E}\left[\frac{\mathbb{I}[i\in W(R,p)]}{|W(R,p)|}\right]
\end{align*}
for $i\in [4]$, where $W(R,p)=\argmax_{j\in[4]}(R_j+p_j)$ and the perturbations are i.i.d. from $\mathcal{N}$. Using this function based on cumulative payoffs, $A_{f^{\text{FPL}}}$ is equivalent to \cref{alg:fpl}. In detail, the probability of choosing $\opt_i$ in the $t^{th}$ round in \cref{alg:fpl} is exactly \begin{align*}
    \mathbb{E}\left[\frac{\mathbb{I}[i\in W(R_{t-1},p_t)]}{|W(R_{t-1},p_t)|}\right].
\end{align*} 
Now we only need to verify that $f^{\text{FPL}}$ satisfies \cref{asm.symf}, \cref{asm.fullsupf} and \cref{asm.consistf}.

Assumption~\ref{asm.symf} holds because the expression $$\mathbb{E}\left[\frac{\mathbb{I}[i\in W(R,p)]}{|W(R,p)|}\right]$$ is permuted in the same way as the reward coordinates, by the i.i.d. perturbations and uniform tie-breaking.

We know that for $\forall \epsilon>0$, we can find a large enough positive constant $n$ such that $\Pr\{p-p'\geq n|p,p'\overset{\mathrm{iid}}{\sim}\mathcal{N}\}<\epsilon$. Therefore, when $R_1>\max\{R_2,R_3,R_4\}+n$, we can deduce that\begin{align*}
    &f^{\text{FPL}}_1(R_1,R_2,R_3,R_4)\\
    &\geq 1-\Pr\left\{\exists i\in\{2,3,4\}:R_1+p_1\leq R_i+p_i\right\}\\
    &\geq 1-\sum_{i=2,3,4}\Pr\left\{R_1+p_1\leq R_i+p_i\middle|p_i\overset{\mathrm{iid}}{\sim} \mathcal{N},i\in[4]\right\}\\
    &\geq 1-\sum_{i=2,3,4}\Pr\left\{p_1+n<p_i\middle|p_i\overset{\mathrm{iid}}{\sim} \mathcal{N},i\in[4]\right\}\\
    &\geq 1-3\epsilon.
\end{align*}
Therefore, Assumption~\ref{asm.fullsupf} holds.
For Assumption~\ref{asm.consistf}, implement uniform tie-breaking by i.i.d. continuous auxiliary priorities independent of the perturbations. If $R_i\geq R_j$, swapping the perturbations and priorities of $i,j$ maps every outcome where $j$ wins to one where $i$ wins. This probability-preserving swap proves $f_i(R)\geq f_j(R)$.
\end{proof}

\subsubsection{Multiplicative Weight Algorithm}
Multiplicative weights algorithm (or hedge algorithm) is first introduced in~\cite{freund1997decision}, which is called {\it exponentially weighted averaged forecaster} by them. Its general no-regret guarantees use suitably tuned learning rates; here the learning rate is fixed. It can be written as Algorithm~\ref{alg:hedge2} for an online decision problem with four options $\opt_i,i\in[4]$.
\begin{algorithm}
    \caption{Multiplicative Weights algorithm in \cite{freund1997decision}.}\label{alg:hedge2}
    \DontPrintSemicolon
    \KwIn{A positive constant $\beta$.}
    Initialize $w_{i,1}=1$ for $i\in[4]$.\\
    \For{$t=1,2,\cdots$}{
        Choose option $\opt_i$ with probability $q_{i,t}=\frac{w_{i,t}}{\sum_{j\in[4]}w_{j,t}}$ in the $t^{th}$ round.\\
        Update $w_{i,t+1}=w_{i,t}\exp(\beta r_{i,t})$.
    }
\end{algorithm}

In formal, we have the following theorem.
\begin{theorem}
Algorithm~\ref{alg:hedge2} is a reward-based online learning algorithm included in $\mathcal{A}$.
\end{theorem}

\begin{proof}
To prove that \cref{alg:hedge2} is in our algorithm family $\mathcal{A}$, we can use a function $f^{\text{hedge 2}}$ such that \begin{align*}
f^{\text{hedge 2}}_i(R_1,R_2,R_3,R_4)=\frac{e^{\beta R_i}}{e^{\beta R_1}+e^{\beta R_2}+e^{\beta R_3}+e^{\beta R_4}},
\end{align*}
for $i\in [4]$. Using this function based on cumulative payoffs, $A_{f^{\text{hedge 2}}}$ is equivalent to \cref{alg:hedge2}. In detail, the probability of choosing $\opt_i$ in the $t^{th}$ round in \cref{alg:hedge2} is exactly \begin{align*}
    \frac{w_i}{w_1+w_2+w_3+w_4}&=\frac{\prod_{j=1}^{t-1}\exp(\beta r_{i,j})}{\sum_{k\in[4]}\prod_{j=1}^{t-1}\exp(\beta r_{k,j})}\\
    &=\frac{e^{\beta R_{i,t-1}}}{e^{\beta R_{1,t-1}}+e^{\beta R_{2,t-1}}+e^{\beta R_{3,t-1}}+e^{\beta R_{4,t-1}}}.
\end{align*} 
Now we only need to verify that $f^{\text{hedge 2}}$ satisfies \cref{asm.symf}, \cref{asm.fullsupf} and \cref{asm.consistf}.

Assumption~\ref{asm.symf} holds because $\frac{e^{\beta R_i}}{e^{\beta R_1}+e^{\beta R_2}+e^{\beta R_3}+e^{\beta R_4}}$ is symmetrical with respect to $R_1,R_2,R_3,R_4$. Assumption~\ref{asm.consistf} holds because $\beta>0$ and exponential function is monotonic. When $R_1-\max\{R_2,R_3,R_4\}>n>0$, we have\begin{align*}
    f_1^{\text{hedge 2}}(R_1,R_2,R_3,R_4)&=\frac{e^{\beta R_1}}{e^{\beta R_1}+e^{\beta R_2}+e^{\beta R_3}+e^{\beta R_4}}\\
    &=\frac{1}{1+e^{\beta (R_2-R_1)}+e^{\beta (R_3-R_1)}+e^{\beta (R_4-R_1)}}\\
    &>\frac{1}{1+3e^{-\beta n}}.
\end{align*}
Therefore, we deduce that \begin{align*}
    \liminf_{R_1-\max\{R_2,R_3,R_4\}\rightarrow+\infty}f_1(R_1,R_2,R_3,R_4)\geq\frac{1}{1+3e^{-\beta n}}\qquad\text{for every }n>0.
\end{align*} 
Letting $n\to\infty$ and using $f_1\leq1$ proves that the limit is one. Therefore, \cref{asm.fullsupf} holds.
\end{proof}

The additional condition in \cref{thm.converge} holds for Hedge with $q(x)=\min\{1,3e^{-\beta x}\}$ and for FTL with uniform tie-breaking with $q=0$. For FPL, a deviation from an option with accumulated payoff difference $x>0$ requires $p_j-p_i\geq x$ for some other option, so one may take
\[
    q(x)=\min\{1,3\Pr(p_2-p_1\geq x)\}.
\]
This is integrable whenever $\mathbb{E}|p_1|<\infty$, since its integral is at most $3\mathbb{E}[(p_2-p_1)_+]$. In particular, bounded noise distributions satisfy the condition; arbitrary noise distributions need not. Exponential bounds and polynomial bounds $C(1+x)^{-1-\eta}$ with $\eta>0$ are further examples of nonincreasing integrable bounds.

\subsection{Justifications for Assumption~\ref{asm.fullsupf}}
\label{app.justasm}

In this section, we justify \cref{asm.fullsupf} in the case where the displayed limit exists: a limit strictly below one precludes no regret.

\begin{theorem}
For a reward-based function $f:\mathbb{R}^4\rightarrow\triangle^3$, if $$\lim_{R_1-\max\{R_2,R_3,R_4\}\to+\infty}f_1(R_1,R_2,R_3,R_4)=1-c<1,$$
mechanism $A_f$ cannot be a no-regret algorithm for general online decision problem. 
\end{theorem}

\begin{proof}
We design an online decision problem such that $r_{1,t}=1,r_{2,t}=r_{3,t}=r_{4,t}=0$ for $t\in\mathbb{N}^+$. According to the description of $f$, there exists a $n$ such that when $R_1-\max\{R_2,R_3,R_4\}\geq n$, $f_1(R_1,R_2,R_3,R_4)<1-\frac{c}{2}$. Notice that when $t>n+1$, we have $R_{1,t-1}=t-1>n=\max\{R_{2,t-1},R_{3,t-1},R_{4,t-1}\}+n$, therefore, when $t>n+1$, the agent chooses $\opt_1$ with probability at most $1-\frac{c}{2}$ in the $t^{th}$ round.

Therefore, we can deduce that when $T>n+2$, $\mathbb{E}[Reg(T)]\geq \frac{c}{2}(T-n-1)$, which is a linear function of $T$. Therefore, $A_f$ is not no-regret.
\end{proof}

\section{More Algorithms in \texorpdfstring{$\mathcal{A}$}{Lg} Applying on CA Mechanism Binary Sequential Peer Prediction}
\label{app.morealgina}

In \cref{app.algina}, we have introduced two widely used algorithms that are reward-based online learning algorithms. In this section, we show that there are even more existing learning algorithms contained by $\mathcal{A}$ when the game is exactly binary sequential peer prediction using CA mechanism. 

\subsection{Replicator Dynamics}
Replicator dynamics track a set of agents in a repeating game and each agent chooses a pure strategy with a probability proportional to expected payoffs deviating to higher-payoff options. We use a similar implementation as \cite{shnayder2016measuring} in a general discrete form here to show that replicator dynamics are also in $\mathcal{A}$ for the binary signal peer prediction problem. To be more specific, during the repeating game, the agent maintains four probabilities $q_{i,t},i\in[4]$ and choose $\opt_i$ with probability $q_{i,t}$, and then update them in the end of the $t^{th}$ round. The updating rule is set as below:\begin{align}
    q_{i,t+1}=\frac{h(r_{i,t})q_{i,t}}{\sum_{j\in[4]}h(r_{j,t})q_{j,t}},i\in[4],\label{eq.updruleofrd}
\end{align}
where $h:\mathbb{R}\rightarrow \mathbb{R}^+$ is a monotonic function satisfying $h(1)>h(-1)$; initialize $q_{i,1}=1/4$ for all $i\in[4]$. Due to the variety of $h$, our discretized replicator dynamics contain extensive learning algorithms. Common discretization of replicator dynamics set $h$ as exponential function or linear function, but we consider general $h$ functions here.

\begin{algorithm}
    \caption{Replicator dynamics.}\label{alg:replicatordynamics}
    \DontPrintSemicolon
    \KwIn{Monotonic function $h:\mathbb{R}\rightarrow \mathbb{R}^+$ satisfying $h(1)>h(-1)$.}
    Initialize $q_{i,1}=1/4$ for $i\in[4]$.\\
    \For{$t=1,2,\cdots$}{
        Choose option $\opt_i$ with probability $q_{i,t}$ in the $t^{th}$ round.\\
        Set $q_{i,t+1}=\frac{h(r_{i,t})q_{i,t}}{\sum_{j\in[4]}h(r_{j,t})q_{j,t}},i\in[4]$.
    }
\end{algorithm}

We have the following theorem.
\begin{theorem}\label{thm.replicatordynamicsina}
When applying on binary sequential peer prediction using CA mechanism, \cref{alg:replicatordynamics} is included in $\mathcal{A}$.
\end{theorem}

\begin{proof}
According to \cref{eq.updruleofrd}, we know that\begin{align}
    \frac{q_{1,t+1}}{q_{2,t+1}}&=\prod_{i=1}^t\frac{h(r_{1,i})}{h(r_{2,i})}\nonumber\\
    &=\left(\frac{h(1)}{h(-1)}\right)^{\sum_{i\in[t]}\mathbb{I}[r_{1,i}=1\wedge r_{2,i}=-1]}\times \left(\frac{h(-1)}{h(1)}\right)^{\sum_{i\in[t]}\mathbb{I}[r_{1,i}=-1\wedge r_{2,i}=1]}\label{eq.expansionoffrac}\\
    &=\left(\frac{h(1)}{h(-1)}\right)^{\frac{R_{1,t}-R_{2,t}}{2}}.\label{eq.resultoffrac}
\end{align}
Here, \cref{eq.expansionoffrac} and \cref{eq.resultoffrac} are because there are only three realizations of $(r_{1,i},r_{2,i})$, which are $(1,-1),(0,0)$ and $(-1,1)$ by the definition of CA mechanism and binary sequential peer prediction. Similarly for $\opt_3,\opt_4$, we can deduce that $\frac{q_{3,t+1}}{q_{4,t+1}}=\left(\frac{h(1)}{h(-1)}\right)^{\frac{R_{3,t}-R_{4,t}}{2}}$. Moreover, we have\begin{align*}
    \frac{q_{1,t+1}q_{2,t+1}}{q_{3,t+1}q_{4,t+1}}&=\frac{h(r_{1,t})h(r_{2,t})}{h(r_{3,t})h(r_{4,t})}\frac{q_{1,t}q_{2,t}}{q_{3,t}q_{4,t}}\\
    &=\frac{q_{1,t}q_{2,t}}{q_{3,t}q_{4,t}}=\frac{q_{1,t-1}q_{2,t-1}}{q_{3,t-1}q_{4,t-1}}=\cdots=1.
\end{align*}
This is because $(r_{1,i},r_{2,i},r_{3,i},r_{4,i})$ has only five possible realizations, which are $(0,0,0,0)$, $(1,-1,1,-1)$, $(1,-1,-1,1)$, $(-1,1,1,-1)$ and $(-1,1,-1,1)$ according to the definition of CA mechanism and binary sequential peer prediction for any $i\in\mathbb{N}^+$.

According to \cref{lemma.sumofpayoffs}, we know that $R_{1,t}+R_{2,t}=R_{3,t}+R_{4,t}$. Therefore, we can deduce that for replicator dynamics applying on binary sequential peer prediction using CA mechanism, we have\begin{align*}
    &q_{1,t+1}=\frac{\left(\frac{h(1)}{h(-1)}\right)^{\frac{R_{1,t}}{2}}}{\left(\frac{h(1)}{h(-1)}\right)^{\frac{R_{1,t}}{2}}+\left(\frac{h(1)}{h(-1)}\right)^{\frac{R_{2,t}}{2}}+\left(\frac{h(1)}{h(-1)}\right)^{\frac{R_{3,t}}{2}}+\left(\frac{h(1)}{h(-1)}\right)^{\frac{R_{4,t}}{2}}},\\
    &q_{2,t+1}=\frac{\left(\frac{h(1)}{h(-1)}\right)^{\frac{R_{2,t}}{2}}}{\left(\frac{h(1)}{h(-1)}\right)^{\frac{R_{1,t}}{2}}+\left(\frac{h(1)}{h(-1)}\right)^{\frac{R_{2,t}}{2}}+\left(\frac{h(1)}{h(-1)}\right)^{\frac{R_{3,t}}{2}}+\left(\frac{h(1)}{h(-1)}\right)^{\frac{R_{4,t}}{2}}},\\
    &q_{3,t+1}=\frac{\left(\frac{h(1)}{h(-1)}\right)^{\frac{R_{3,t}}{2}}}{\left(\frac{h(1)}{h(-1)}\right)^{\frac{R_{1,t}}{2}}+\left(\frac{h(1)}{h(-1)}\right)^{\frac{R_{2,t}}{2}}+\left(\frac{h(1)}{h(-1)}\right)^{\frac{R_{3,t}}{2}}+\left(\frac{h(1)}{h(-1)}\right)^{\frac{R_{4,t}}{2}}},\\
    &q_{4,t+1}=\frac{\left(\frac{h(1)}{h(-1)}\right)^{\frac{R_{4,t}}{2}}}{\left(\frac{h(1)}{h(-1)}\right)^{\frac{R_{1,t}}{2}}+\left(\frac{h(1)}{h(-1)}\right)^{\frac{R_{2,t}}{2}}+\left(\frac{h(1)}{h(-1)}\right)^{\frac{R_{3,t}}{2}}+\left(\frac{h(1)}{h(-1)}\right)^{\frac{R_{4,t}}{2}}}.
\end{align*}
Hence, now replicator dynamics behave the same as mechanism $A_{f^{\text{replicator}}}$ such that\begin{align*}
    &f^{{\text{replicator}}}_1(R_1,R_2,R_3,R_4)=\frac{\left(\frac{h(1)}{h(-1)}\right)^{\frac{R_{1}}{2}}}{\left(\frac{h(1)}{h(-1)}\right)^{\frac{R_{1}}{2}}+\left(\frac{h(1)}{h(-1)}\right)^{\frac{R_{2}}{2}}+\left(\frac{h(1)}{h(-1)}\right)^{\frac{R_{3}}{2}}+\left(\frac{h(1)}{h(-1)}\right)^{\frac{R_{4}}{2}}},\\
    &f^{{\text{replicator}}}_2(R_1,R_2,R_3,R_4)=\frac{\left(\frac{h(1)}{h(-1)}\right)^{\frac{R_{2}}{2}}}{\left(\frac{h(1)}{h(-1)}\right)^{\frac{R_{1}}{2}}+\left(\frac{h(1)}{h(-1)}\right)^{\frac{R_{2}}{2}}+\left(\frac{h(1)}{h(-1)}\right)^{\frac{R_{3}}{2}}+\left(\frac{h(1)}{h(-1)}\right)^{\frac{R_{4}}{2}}},\\
    &f^{{\text{replicator}}}_3(R_1,R_2,R_3,R_4)=\frac{\left(\frac{h(1)}{h(-1)}\right)^{\frac{R_{3}}{2}}}{\left(\frac{h(1)}{h(-1)}\right)^{\frac{R_{1}}{2}}+\left(\frac{h(1)}{h(-1)}\right)^{\frac{R_{2}}{2}}+\left(\frac{h(1)}{h(-1)}\right)^{\frac{R_{3}}{2}}+\left(\frac{h(1)}{h(-1)}\right)^{\frac{R_{4}}{2}}},\\
    &f^{{\text{replicator}}}_4(R_1,R_2,R_3,R_4)=\frac{\left(\frac{h(1)}{h(-1)}\right)^{\frac{R_{4}}{2}}}{\left(\frac{h(1)}{h(-1)}\right)^{\frac{R_{1}}{2}}+\left(\frac{h(1)}{h(-1)}\right)^{\frac{R_{2}}{2}}+\left(\frac{h(1)}{h(-1)}\right)^{\frac{R_{3}}{2}}+\left(\frac{h(1)}{h(-1)}\right)^{\frac{R_{4}}{2}}}.
\end{align*} 
It is easy to verify that $f$ satisfies \cref{asm.symf}, \cref{asm.fullsupf} and \cref{asm.consistf}.

In detail, \cref{asm.symf} holds because $f^{{\text{replicator}}}$ is symmetrical obviously with respect to $R_1,R_2,R_3,R_4$. Moreover, we know that\begin{align*}
    f^{{\text{replicator}}}_1(R_1,R_2,R_3,R_4)&\geq \frac{\left(\frac{h(1)}{h(-1)}\right)^{\frac{R_{1}}{2}}}{\left(\frac{h(1)}{h(-1)}\right)^{\frac{R_{1}}{2}}+3\left(\frac{h(1)}{h(-1)}\right)^{\frac{\max\left\{R_2,R_3,R_4\right\}}{2}}}\\
    &=\frac{1}{1+3\left(\frac{h(1)}{h(-1)}\right)^{\frac{\max\left\{R_2,R_3,R_4\right\}-R_{1}}{2}}}.
\end{align*}
Therefore, when $R_1-\max\left\{R_2,R_3,R_4\right\}\rightarrow+\infty$, we have $\frac{1}{1+3\left(\frac{h(1)}{h(-1)}\right)^{\frac{\max\left\{R_2,R_3,R_4\right\}-R_{1}}{2}}}\rightarrow 1$. Thus $f^{{\text{replicator dynamics}}}_1(R_1,R_2,R_3,R_4)\rightarrow 1$. Hence, we have proved that \cref{asm.fullsupf} holds. Finally, \cref{asm.consistf} also holds obviously according to the definition of $f^{{\text{replicator}}}$ and monotonicity of exponential function.
\end{proof}

\subsection{An Alternating Version of Multiplicative Weights Algorithm}

In \cref{app.algina}, we have already introduced an exponential updating function form of multiplicative weights algorithm \cite{freund1997decision}. There is another version multiplicative weights algorithm introduced in the survey \cite{arora2012multiplicative}, which can be written as \cref{alg:hedge1} for our particular binary sequential peer prediction problem using CA mechanism. 

\begin{algorithm}
    \caption{Multiplicative Weights algorithm in \cite{arora2012multiplicative}.}\label{alg:hedge1}
    \DontPrintSemicolon
    \KwIn{A constant $0<\beta<1$.}
    Initialize $w_{i,1}=1$ for $i\in[4]$.\\
    \For{$t=1,2,\cdots$}{
        Choose option $\opt_i$ with probability $q_{i,t}=\frac{w_{i,t}}{\sum_{j\in[4]}w_{j,t}}$ in the $t^{th}$ round.\\
        Update $w_{i,t+1}=\begin{cases}(1+\beta)w_{i,t}&r_{i,t}=1\\w_{i,t}&r_{i,t}=0\\(1-\beta)w_{i,t}&r_{i,t}=-1\end{cases}$.
    }
\end{algorithm}
We only need to set $h:\mathbb{R}\rightarrow\mathbb{R}^+$ to satisfy that $h(1)=1+\beta$, $h(0)=1$ and $h(-1)=1-\beta$. Then \cref{alg:replicatordynamics} behaves completely the same as \cref{alg:hedge1}. Therefore, according to \cref{thm.replicatordynamicsina}, we have \cref{thm.hedge1ina} for this alternating version of multiplicative weights algorithm. 

\begin{theorem}\label{thm.hedge1ina}
When applying on binary sequential peer prediction using CA mechanism, \cref{alg:hedge1} is included in $\mathcal{A}$.
\end{theorem}

\section{Proofs and Details of Section~\ref{sec.main}}
\label{app.mainproof}

\subsection{Proofs of Properties of CA Mechanism in Binary Signal Peer Prediction Games}
\label{app.propertyca}

\subsubsection{Proof of Lemma~\ref{lemma.boundeduninformative}}

\begin{proof}
By the definition of CA mechanism, we know that $r_{3,j}=\mathbb{I}[1=\hat{y}_j]-\mathbb{I}[1=\hat{y}_{j-1}]$ for $j\in\mathbb{N}^+$. Therefore, we can deduce that\begin{align*}
    R_{3,t}&=\sum_{j=1}^tr_{3,j}\\
    &=\sum_{j=1}^t\left(\mathbb{I}[1=\hat{y}_j]-\mathbb{I}[1=\hat{y}_{j-1}]\right)\\
    &=\sum_{j=1}^t\mathbb{I}[1=\hat{y}_j]-\sum_{j=0}^{t-1}\mathbb{I}[1=\hat{y}_j]\\
    &=\mathbb{I}[1=\hat{y}_t]-\mathbb{I}[1=\hat{y}_0]\in[-1,1].
\end{align*}
For strategy $\sum_{j=1}^tp_{4,j},\sum_{j=1}^tq_{3,j},\sum_{j=1}^tq_{4,j}$, the deductions are all similar.
\end{proof}

\subsubsection{Proof of Lemma~\ref{lemma.sumofpayoffs}}

\begin{proof}
By the definition of CA mechanism, we know that $r_{1,j}=\mathbb{I}[x_j=\hat{y}_j]-\mathbb{I}[x_j=\hat{y}_{j-1}]$ and $r_{2,j}=\mathbb{I}[1-x_j=\hat{y}_j]-\mathbb{I}[1-x_j=\hat{y}_{j-1}]$. Therefore, we can deduce that\begin{align*}
    R_{1,t}+R_{2,t}&=\sum_{j=1}^t(r_{1,j}+r_{2,j})\\
    &=\sum_{j=1}^t\left(\mathbb{I}[x_j=\hat{y}_j]-\mathbb{I}[x_j=\hat{y}_{j-1}]\right)+\sum_{j=1}^t\left(\mathbb{I}[1-x_j=\hat{y}_j]-\mathbb{I}[1-x_j=\hat{y}_{j-1}]\right)\\
    &=\sum_{j=1}^t(\mathbb{I}[x_j=\hat{y}_j]+\mathbb{I}[1-x_j=\hat{y}_j])-\sum_{j=1}^t(\mathbb{I}[x_j=\hat{y}_{j-1}]+\mathbb{I}[1-x_j=\hat{y}_{j-1}])\\
    &=t-t=0.
\end{align*}
Similar deductions can be made for $R_{3,t}+R_{4,t}$, $S_{1,t}+S_{2,t}$ and $S_{3,t}+S_{4,t}$, which completes the proof.
\end{proof}

\subsection{Proofs of Truthful Convergence for CA mechanism on Reward-based Algorithms}
\fang{}
\subsubsection{Proof of Lemma~\ref{lemma.nonoptimality34}}
\label{app.proofnonoptimality34}

\begin{proof}
Because Alice and Bob are symmetric, we only consider Alice.
By \cref{lemma.sumofpayoffs}, without loss of generality, we suppose that $\sum_{j=1}^tr_{1,j}\geq0$. Then $R_{1,t}\geq 1$ and $R_{2,t}\leq 0$ or $R_{1,t}=R_{2,t}=0$ because each $r_{1,j}$ equals to $-1,0$ or $1$, which is an integer. If $R_{1,t}\geq 1$, then we can deduce that $R_{1,t}\geq \max_{i=2,3,4}R_{i,t}$. Therefore, $f_1(R_{1,t},R_{2,t},R_{3,t},R_{4,t})\geq\max_{i=2,3,4}f_i(R_{1,t},R_{2,t},R_{3,t},R_{4,t})$ according to \cref{asm.consistf}. We know that $\sum_{i\in[4]}f_i(R_{1,t},R_{2,t},R_{3,t},R_{4,t})=1$, so $f_1(R_{1,t},R_{2,t},R_{3,t},R_{4,t})\geq\frac{1}{4}$.

If $R_{1,t}=R_{2,t}=R_{3,t}=R_{4,t}=0$, we can also deduce that $R_{1,t}\geq \max_{i=2,3,4}R_{i,t}$. Therefore, $f_1(R_{1,t},R_{2,t},R_{3,t},R_{4,t})\geq\max_{i=2,3,4}f_i(R_{1,t},R_{2,t},R_{3,t},R_{4,t})$ according to \cref{asm.consistf}. We know that $\sum_{i\in[4]}f_i(R_{1,t},R_{2,t},R_{3,t},R_{4,t})=1$, so $f_1(R_{1,t},R_{2,t},R_{3,t},R_{4,t})\geq\frac{1}{4}$.

Finally, we prove that $\{R_{1,t}=R_{2,t}=0,R_{3,t}=1,R_{4,t}=-1\}$ and $\{R_{1,t}=R_{2,t}=0,R_{3,t}=-1,R_{4,t}=1\}$ never occur in our game. According to the definition of CA mechanism, we know that $(r_{1,t},r_{2,t},r_{3,t},r_{4,t})$ has only five possibilities, which are $(0,0,0,0)$, $(1,-1,1,-1)$, $(1,-1,-1,1)$, $(-1,1,1,-1)$ and $(-1,1,-1,1)$. Therefore, we can deduce that $r_{1,t}-r_{2,t}+r_{3,t}-r_{4,t}$ can be divided by $4$. Hence, $R_{1,t}-R_{2,t}+R_{3,t}-R_{4,t}=\sum_{i=1}^{t}(r_{1,i}-r_{2,i}+r_{3,i}-r_{4,i})$ is also divided by $4$. This indicates that both $\{\sum_{j=1}^tr_{1,j}=\sum_{j=1}^tr_{2,j}=0,\sum_{j=1}^tr_{3,j}=1,\sum_{j=1}^tr_{4,j}=-1\}$ and $\{\sum_{j=1}^tr_{1,j}=\sum_{j=1}^tr_{2,j}=0,\sum_{j=1}^tr_{3,j}=-1,\sum_{j=1}^tr_{4,j}=1\}$ cannot happen.

To sum up, the probability that Alice chooses $\opt_1$ in round $t+1$ is at least $\frac{1}{4}$, which is what we want.
\end{proof}

\subsubsection{Proof of Lemma~\ref{lemma.delta}}
\label{app.proofofdelta}

\begin{proof}
It is easy to verify that $\gamma_1>0$. Moreover, under \cref{asm.poscorr},
\[
    \gamma_1-\gamma_2=4\bigl(P_{X,Y}(1,1)P_{X,Y}(0,0)-P_{X,Y}(1,0)P_{X,Y}(0,1)\bigr)>0.
\]
Choose $0<\delta\leq(\gamma_1-\gamma_2)/8$ and the corresponding $c_1$. Fix a history $\mathcal{H}_{t-1}$ with $S_{1,t-1}>c_1+1$; all expectations and probabilities below are conditional on this history. For $s\geq2$, let
\[
 A_s=\mathbb{I}[x_s=y_s]-\mathbb{I}[1-x_s=y_s]
       -\mathbb{I}[x_s=y_{s-1}]+\mathbb{I}[1-x_s=y_{s-1}].
\]
Writing $\epsilon_1=1$ and $\epsilon_2=-1$, for $i\in\{1,2\}$ the following comparison holds for every realization:
\[
 \left|r_{1,s}-r_{2,s}-\epsilon_i A_s\right|
 \leq2\mathbb{I}[\opt_s^Y\ne\opt_i]+2\mathbb{I}[\opt_{s-1}^Y\ne\opt_i].
\]
Indeed, changing either of Bob's two reports changes the corresponding difference of indicators by at most $2$. Thus no independence between an adaptive choice and previous signals is required. Both deviation probabilities for $s=t+1$ and $i=1$ are at most $\delta$, since $S_{1,t}$ differs from $S_{1,t-1}$ by at most $1$. The fresh i.i.d. signals give $\mathbb{E}[A_{t+1}]=\gamma_1-\gamma_2$. Therefore,\begin{align*}
    \mathbb{E}[r_{1,t+1}-r_{2,t+1}]&=\mathbb{E}\bigl[\mathbb{I}[x_{t+1}=\hat{y}_{t+1}]-\mathbb{I}[x_{t+1}=\hat{y}_{t}]\\
    &\qquad{}-\mathbb{I}[1-x_{t+1}=\hat{y}_{t+1}]+\mathbb{I}[1-x_{t+1}=\hat{y}_{t}]\bigr]\\
    &=\mathbb{E}\left[\mathbb{I}[x_{t+1}=\hat{y}_{t+1}]-\mathbb{I}[1-x_{t+1}=\hat{y}_{t+1}]\right]\\
    &\qquad{}-\mathbb{E}\left[\mathbb{I}[x_{t+1}=\hat{y}_{t}]-\mathbb{I}[1-x_{t+1}=\hat{y}_{t}]\right]\\
    &\geq\mathbb{E}[A_{t+1}]-2\Pr\{\opt_{t+1}^Y\ne\opt_1\}-2\Pr\{\opt_t^Y\ne\opt_1\}\\
    &\geq\gamma_1-\gamma_2-4\delta\geq\frac{\gamma_1-\gamma_2}{2}>0.
\end{align*}
This proves the lemma. The same pointwise comparison with $i=2$ will be used below.
\end{proof}

\subsubsection{Proof of Lemma~\ref{lemma.expectr12}}
\label{app.proofofexpectr12}

\begin{proof}
Noticing that when $\mathcal{E}_t^{1,2}$ happens, Bob will choose $\opt_2$ with probability larger than $1-\delta$ in the next $\lceil\frac{1000}{\gamma_1-\gamma_2}\rceil+1$ rounds. Therefore, we can deduce that\begin{align}
    &\mathbb{E}\left[\sum_{j=1}^{\lceil\frac{1000}{\gamma_1-\gamma_2}\rceil+1}\left(r_{1,t+j}-r_{2,t+j}\right)\middle|\mathcal{H}_t\text{ such that }\mathcal{E}_t^{1,2}=1\right]\nonumber\\
    &=\sum_{j=1}^{\lceil\frac{1000}{\gamma_1-\gamma_2}\rceil+1}(\mathbb{E}\left[\mathbb{I}[x_{t+j}=\hat{y}_{t+j}]-\mathbb{I}[1-x_{t+j}=\hat{y}_{t+j}]\middle|\mathcal{H}_t\text{ such that }\mathcal{E}_t^{1,2}=1\right]\nonumber\\
    &-\mathbb{E}\left[\mathbb{I}[x_{t+j}=\hat{y}_{t+j-1}]-\mathbb{I}[1-x_{t+j}=\hat{y}_{t+j-1}]\middle|\mathcal{H}_t\text{ such that }\mathcal{E}_t^{1,2}=1\right])\nonumber\\
    &\leq 2+\sum_{j=2}^{\lceil\frac{1000}{\gamma_1-\gamma_2}\rceil+1}\left(-\gamma_1+\gamma_2+4\delta\right)\label{eq.relaxexpectations}\\
    &\leq 2-\sum_{j=2}^{\lceil\frac{1000}{\gamma_1-\gamma_2}\rceil+1}\frac{\gamma_1-\gamma_2}{2}\label{eq.removedelta}\\
    &<-100.\nonumber
\end{align}
Here, \cref{eq.relaxexpectations} follows from the pointwise comparison in the proof of \cref{lemma.delta}, with $i=2$. For $j\geq2$, $\mathbb{E}[A_{t+j}\mid\mathcal{H}_t]=\gamma_1-\gamma_2$, while both deviation indicators have conditional expectations at most $\delta$. This follows by conditioning first on each preceding history and then on $\mathcal{H}_t$. The term for $j=1$ is bounded by $2$. Finally, \cref{eq.removedelta} follows from $4\delta\leq(\gamma_1-\gamma_2)/2$.
\end{proof}

\subsubsection{Proof of Lemma~\ref{lemma.alwaysnicesituation}}
\label{app.proofalwaysnicesituation}

\begin{proof}
First we show the complement of event $\limsup_{t\to \infty}\overline{\mathcal{E}_t^{1,2}\vee\mathcal{E}_t^{2,1}}$ is 
\begin{equation}\label{eq:alwaysnicesituation0}
    \liminf_{t\to \infty} \mathcal{E}_{t}^{1,2}\vee \liminf_{t\to \infty} \mathcal{E}_{t}^{2,1}.
\end{equation}
According to the definition of $\limsup$, we can write $\limsup_{t\rightarrow\infty}\overline{\mathcal{E}_t^{1,2}\vee\mathcal{E}_t^{2,1}}=1$ as $\wedge_{t=1}^{\infty}\left(\vee_{i=t}^\infty \overline{\mathcal{E}_i^{1,2}\vee\mathcal{E}_i^{2,1}}\right)=1$. This is equivalent to $\vee_{t=1}^{\infty}\left(\wedge_{i=t}^\infty \left(\mathcal{E}_i^{1,2}\vee\mathcal{E}_i^{2,1}\right)\right)=0$ by De Morgan's law, which can be written as $\liminf_{t\rightarrow\infty}\left(\mathcal{E}_t^{1,2}\vee\mathcal{E}_t^{2,1}\right)=0$.

More concretely, according to the definition of $\liminf$, the complement of $\limsup_{t\rightarrow\infty}\overline{\mathcal{E}_t^{1,2}\vee\mathcal{E}_t^{2,1}}=1$ is \begin{align*}
    \left\{\exists T\in\mathbb{N}^+,\forall t\geq T,\mathcal{E}_t^{1,2}\vee\mathcal{E}_t^{2,1}=1\right\}.
\end{align*}
Therefore, we only need to prove that $\Pr\left\{\exists T\in\mathbb{N}^+,\forall t\geq T,\mathcal{E}_t^{1,2}\vee\mathcal{E}_t^{2,1}=1\right\}=0$ which is the event in \cref{eq:alwaysnicesituation0}. We denote all the game history through the $t^{th}$ round as $\mathcal{H}_t$. Then we only need to prove that for any $\mathcal{H}_T$, the conditional probability $\Pr\left\{\mathcal{E}_{t}^{1,2}\vee \mathcal{E}_{t}^{2,1}=1,\forall t\geq T\middle|\mathcal{H}_T\right\}=0$. This is because \begin{align*}
    \Pr\left\{\exists T\in\mathbb{N}^+,\forall t\geq T,\mathcal{E}_t^{1,2}\vee\mathcal{E}_t^{2,1}=1\right\}&\leq\sum_{T=1}^\infty\sum_{\mathcal{H}_T}\Pr\{\mathcal{H}_T\}\Pr\left\{\forall t\geq T,\mathcal{E}_{t}^{1,2}\vee \mathcal{E}_{t}^{2,1}=1\middle|\mathcal{H}_T\right\},
\end{align*}
where the number of summed terms is countable and we know that the sum of countable infinite zeros is still zero.

Moreover, if $\mathcal{E}_t^{1,2}=1$, $\mathcal{E}_{t+1}^{2,1}\neq 1$ according to the definition, we know that $\left\{\mathcal{E}_t^{1,2}=1,\mathcal{E}_{t+1}^{2,1}=1\right\}$ and $\left\{\mathcal{E}_t^{2,1}=1,\mathcal{E}_{t+1}^{1,2}=1\right\}$ always equals to zero. Therefore, $\left\{\forall t\geq T,\mathcal{E}_{t}^{1,2}\vee \mathcal{E}_{t}^{2,1}=1\middle|\mathcal{H}_T\right\}$ is equivalent to \begin{small}\begin{align}
    \left\{\forall t\geq T,\mathcal{E}_{t}^{1,2}\vee \mathcal{E}_{t}^{2,1}=1\middle|\mathcal{H}_T\right\}\wedge\left(\wedge_{t=1}^\infty \overline{\left\{\mathcal{E}_t^{1,2}=1,\mathcal{E}_{t+1}^{2,1}=1\right\}}\right)\wedge\left(\wedge_{t=1}^\infty \overline{\left\{\mathcal{E}_t^{2,1}=1,\mathcal{E}_{t+1}^{1,2}=1\right\}}\right).\label{eq.eqevent}
\end{align}\end{small}
Moreover, suppose the event expressed as \cref{eq.eqevent} happens, we can deduce that for $t\geq T$, if $\mathcal{E}_t^{1,2}=1$, $\mathcal{E}_{t+1}^{2,1}=0$ because  $\left\{\mathcal{E}_t^{1,2}=1,\mathcal{E}_{t+1}^{2,1}=1\right\}=0$, so $\mathcal{E}_{t+1}^{1,2}$ still happens; if $\mathcal{E}_t^{2,1}=1$, $\mathcal{E}_{t+1}^{1,2}=0$ because  $\left\{\mathcal{E}_t^{2,1}=1,\mathcal{E}_{t+1}^{1,2}=1\right\}=0$, so $\mathcal{E}_{t+1}^{2,1}$ still happens. Thus, event in \cref{eq.eqevent} leads to $\left\{\forall t\geq T,\mathcal{E}_{t}^{1,2}=1\middle|\mathcal{H}_T\right\}\vee\left\{ \forall t\geq T,\mathcal{E}_{t}^{2,1}=1\middle|\mathcal{H}_T\right\}$.
Conversely, $\left\{\forall t\geq T, \mathcal{E}_{t}^{1,2}=1\middle|\mathcal{H}_T\right\}\vee\left\{\forall t\geq T, \mathcal{E}_{t}^{2,1}=1\middle|\mathcal{H}_T\right\}$ is obviously included in the event expressed as \cref{eq.eqevent}.
Hence, event $\left\{\forall t\geq T,\mathcal{E}_{t}^{1,2}\vee \mathcal{E}_{t}^{2,1}=1\middle|\mathcal{H}_T\right\}$ is equivalent to \begin{align*}
    \left\{\forall t\geq T,\mathcal{E}_{t}^{1,2}=1\middle|\mathcal{H}_T\right\}\vee\left\{ \forall t\geq T,\mathcal{E}_{t}^{2,1}=1|\mathcal{H}_T\right\}.
\end{align*}
Therefore, without loss of generality, we only need to prove that $\Pr\left\{\forall t\geq T,\mathcal{E}_{t}^{1,2}=1\middle|\mathcal{H}_T\right\}=0$ for any game history $\mathcal{H}_T$.  If $\mathcal{E}^{1,2}_T=0$, the desired conditional probability is already zero, so assume $\mathcal{E}^{1,2}_T=1$. All probabilities below are conditional on this initial history. We use Borel-Cantelli lemma (\cref{thm.borel}) and Azuma-Hoeffding inequality (\cref{thm:asuma}) on the sub-sequence $\left\{t_i:=T+i\left(\lceil\frac{1000}{\gamma_1-\gamma_2}\rceil+1\right):{i\ge 1}\right\}$ to rule out $\{\forall t\geq T,\mathcal{E}^{1,2}_t=1\}$. For simplicity, we denote $\lceil\frac{1000}{\gamma_1-\gamma_2}\rceil$ by $\zeta$.

Let $\tau$ be a stopping time such that $\tau=\inf_{t>T}\{R_{1,t}\leq c_0\text{ or }S_{2,t}\leq c_0\}$ and $i_\tau=\lceil\frac{\tau-T}{\zeta+1}\rceil$, with $\inf\varnothing=\infty$.

We design a series of new variables $\{D_i\}_{i=0,1,2,\cdots}$, where \begin{align*}
    D_i&=R_{1,T+(i\wedge i_\tau)(\zeta+1)}-R_{2,T+(i\wedge i_\tau)(\zeta+1)}+100(i\wedge i_\tau).
\end{align*}

Now we show that $\{D_i\}_{i=0,1,2,\cdots}$ is a super-martingale with bounded difference, and we will use Azuma-Hoeffding inequality (\cref{thm:asuma}) to show the value of $D_i$ (and thus $R_{1,t}$) cannot be too big. Therefore, $\{D_i>2c_0+100i\}$ can only happen finitely many times by Borel-Cantelli lemma (\cref{thm.borel}), which will rule out permanent occurrence of $\mathcal{E}^{1,2}_t$.

Actually, if $\mathcal{H}_{T+i\left(\zeta+1\right)}$ satisfies that $\mathbb{I}[T+i\left(\zeta+1\right)<\tau]=1$, we have $\mathcal{E}^{1,2}_{T+i\left(\zeta+1\right)}=1$. According to \cref{lemma.expectr12}, we have\begin{align*}
    &\mathbb{E}\left[D_{i+1}\middle|\mathcal{H}_{T+i\left(\zeta+1\right)}\right]\\
    &=\mathbb{E}\left[\sum_{j=1}^{\zeta+1}(r_{1,T+i\left(\zeta+1\right)+j}-r_{2,T+i\left(\zeta+1\right)+j})\middle|\mathcal{H}_{T+i\left(\zeta+1\right)}\right]+D_{i}+100\\
    &\leq -100+D_{i}+100=D_i.
\end{align*}
On the other hand, if $\mathcal{H}_{T+i\left(\zeta+1\right)}$ satisfies that $\mathbb{I}[T+i\left(\zeta+1\right)<\tau]=0$, we know that $D_{i+1}=D_i$.

Moreover, when $i\leq i_\tau$, we know that $|D_{i}-D_{i-1}|$ is bounded by $102+2\zeta$ for any $i\in\mathbb{N}^+$ because $|R_{1,t+j}-R_{1,t}|$ is bounded by $j$ by definition for any $t,j\in\mathbb{N}$. When $i>i_\tau$, we know that $D_i=D_{i-1}$.

It is worthy to notice that $\mathcal{E}_{T+i\left(\zeta+1\right)}^{1,2}=1$ implies $R_{1,T+i\left(\zeta+1\right)}>c_0$. Therefore, if $\mathcal{E}_{T+j}^{1,2}=1$ for all $j>T$, $\tau=\infty$, so we have $D_{i}>2c_0+100i$ for every $i$. However, according to Azuma--Hoeffding inequality (\cref{thm:asuma}), for $i>\frac{R_{1,T}-R_{2,T}-2c_0}{100}$, we have\begin{align*}
    \Pr\{D_{i}>2c_0+100i\}&=\Pr\{D_{i}-D_0>2c_0+100i-(R_{1,T}-R_{2,T})\}\\
    &\leq \exp\left(\frac{-(2c_0+100i-(R_{1,T}-R_{2,T}))^2}{2i(102+2\zeta)^2}\right).
\end{align*}
This upper bound of $\Pr\{D_{i}>2c_0+100i\}$ decays exponentially in $i$. By the Borel-Cantelli lemma (\cref{thm.borel}), the event $\{D_i>2c_0+100i\}$ will occur only finitely often almost surely. This is contradictory to $\mathcal{E}_{T+j}^{1,2}=1$ for all $j>T$, so $\Pr\{\mathcal{E}_{t}^{1,2}=1,\forall t\geq T|\mathcal{H}_T\}=0$.
\end{proof}

\subsubsection{Proof of Lemma~\ref{lemma.inflargegap}}
\label{app.proofinflargegap}

\begin{proof}
By symmetry, we only consider the case that $\mathcal{H}_T$ with $R_{1,T}+S_{1,T}\geq 0$ and $R_{1,T}\geq S_{1,T}$. Interchanging the agents handles the reverse ordering. If the sum is negative, complement both private signals and interchange $\opt_1,\opt_2$; this preserves \cref{asm.apriori,asm.poscorr}, the minimum signal probability and the update rules by exchangeability. If $R_{1,T}\leq c_0$, then we have $S_{1,T}\geq -R_{1,T}\geq-c_0$; if $R_{1,T}>c_0$, because $\overline{\mathcal{E}_T^{1,2}\vee\mathcal{E}_T^{2,1}}=1$, $S_{2,T}\leq c_0$, so $S_{1,T}=-S_{2,T}\geq -c_0$. Therefore, we have $R_{1,T}\geq 0$ together with $S_{1,T}\geq -c_0$.

We now propose a process $\mathcal{P}_T$ with less than $T+4(u+c_0)+100$ of rounds, and we will prove that it happens with probability no less than a constant $\lambda_u$ given any $\mathcal{H}_T$ as we suppose. The process $\mathcal{P}_T$ is defined as \begin{enumerate}
    \item If $S_{1,T}\geq0$, skip this phase. Otherwise, $x_{j+1}=1-\hat{y}_j,y_{j+1}=1-\hat{x}_j$ from $j=T$, Alice uses strategy $\opt_1$, Bob uses strategy $\opt_2$ until some round $T_1$ such that $S_{1,T_1}\geq 0$.
    \item Alice and Bob uses strategy $\opt_1$ for $4u+50$ rounds and signals are generated as $x_{j+1}=y_{j+1}=1-x_j$ for $j\geq T_1$.
\end{enumerate}

Firstly, we prove that $\mathcal{P}_T$ will stop in $T+4(u+c_0)+100$ rounds and in the round $T_2$ exactly after $\mathcal{P}_T$, the game will enter good events such that $\mathcal{E}^{1,1}(u)\vee\mathcal{E}^{2,2}(u)=1$. This part can be proved simply according to the definition of $\mathcal{P}_T$.

During this process, we claim that $R_{1,i}$ and $S_{1,i}$ are monotone. If $i\in [T+1, T_1]$, we have\begin{align*}
    &R_{1,i}-R_{1,i-1}=r_{1,i}=\mathbb{I}[x_i=\hat{y}_i]-\mathbb{I}[x_i=\hat{y}_{i-1}]=\mathbb{I}[x_i=\hat{y}_i]\geq 0,\\
    &S_{1,i}-S_{1,i-1}=s_{1,i}=\mathbb{I}[y_i=\hat{x}_i]-\mathbb{I}[y_i=\hat{x}_{i-1}]=\mathbb{I}[y_i=\hat{x}_i]\geq 0.
\end{align*}
Therefore, $R_{1,T_1}\geq R_{1,T}\geq 0$ so after the first phase, both $R_{1,T_1}$ and $S_{1,T_1}$ are non-negative.

If $i\in [T_1+1, T_2]$ where $T_2=T_1+4u+50$, we have\begin{align*}
    &R_{1,i}-R_{1,i-1}=r_{1,i}=\mathbb{I}[x_i=\hat{y}_i]-\mathbb{I}[x_i=\hat{y}_{i-1}]\geq\mathbb{I}[x_i=y_i]-1=0,\\
    &S_{1,i}-S_{1,i-1}=s_{1,i}=\mathbb{I}[y_i=\hat{x}_i]-\mathbb{I}[y_i=\hat{x}_{i-1}]\geq\mathbb{I}[y_i=x_i]-1=0.
\end{align*}

Now we prove that $T_1\leq T+4c_0+50$. Before the first phase stops, $S_{1,i}$ is nondecreasing and gains $2$ in every complete block of four rounds. Since $-S_{1,T}$ is a positive integer at most $c_0$ whenever this phase is needed, this implies $T_1\leq T+4c_0$, as required. Actually, we have\begin{align*}
    S_{1,i+4}-S_{1,i}&=s_{1,i+1}+s_{1,i+2}+s_{1,i+3}+s_{1,i+4}\\
    &=\mathbb{I}[y_{i+1}=\hat{x}_{i+1}]-\mathbb{I}[y_{i+1}=\hat{x}_{i}]+\mathbb{I}[y_{i+2}=\hat{x}_{i+2}]-\mathbb{I}[y_{i+2}=\hat{x}_{i+1}]\\
    &+\mathbb{I}[y_{i+3}=\hat{x}_{i+3}]-\mathbb{I}[y_{i+3}=\hat{x}_{i+2}]+\mathbb{I}[y_{i+4}=\hat{x}_{i+4}]-\mathbb{I}[y_{i+4}=\hat{x}_{i+3}]\\
    &=\mathbb{I}[1-\hat{y}_{i+1}=\hat{x}_{i+1}]+\mathbb{I}[1-\hat{x}_{i+1}=1-\hat{y}_{i+1}]\\
    &\quad{}+\mathbb{I}[\hat{y}_{i+1}=1-\hat{x}_{i+1}]+\mathbb{I}[\hat{x}_{i+1}=\hat{y}_{i+1}]\\
    &=2,
\end{align*}
which is what we want to prove. An example of the first phase is shown as \cref{tab.examplephase1}.

Next, we prove that after the second phase, $\mathcal{E}^{1,1}_{T_2}(u)=1$. Actually, we have for any $i\in[T_1+2,T_2]$,\begin{align*}
    &R_{1,i}-R_{1,i-1}=r_{1,i}=\mathbb{I}[x_i=\hat{y}_i]-\mathbb{I}[x_i=\hat{y}_{i-1}]=\mathbb{I}[x_i=y_i]-\mathbb{I}[x_i=x_{i-1}]=1,\\
    &S_{1,i}-S_{1,i-1}=s_{1,i}=\mathbb{I}[y_i=\hat{x}_i]-\mathbb{I}[y_i=\hat{x}_{i-1}]=\mathbb{I}[y_i=x_i]-\mathbb{I}[y_i=x_{i-1}]=1.
\end{align*}
Therefore, $R_{1,T_2}\geq R_{1,T_1+1}+4u+49\geq R_{1,T_1}+4u+49>u$ and $S_{1,T_2}\geq S_{1,T_1+1}+4u+49\geq S_{1,T_1}+4u+49>u$. Therefore, $\mathcal{E}^{1,1}_{T_2}(u)=1$. An example of the second phase is shown as \cref{tab.examplephase2}.

\begin{table}[!htbp]
\centering
\begin{tabular}{|c|c|c|c|c|c|c|}
\hline
$\hat{x}_{T}=0$&$\hat{x}_{T+1}=0$&$\hat{x}_{T+2}=1$&$\hat{x}_{T+3}=1$&$\hat{x}_{T+4}=0$&$\hat{x}_{T+5}=0$&$\hat{x}_{T+6}=1$\\
\hline
$x_T$&$x_{T+1}=0$&$x_{T+2}=1$&$x_{T+3}=1$&$x_{T+4}=0$&$x_{T+5}=0$&$x_{T+6}=1$\\
\hline
$y_T$&$y_{T+1}=1$&$y_{T+2}=1$&$y_{T+3}=0$&$y_{T+4}=0$&$y_{T+5}=1$&$y_{T+6}=1$\\
\hline
$\hat{y}_{T}=1$&$\hat{y}_{T+1}=0$&$\hat{y}_{T+2}=0$&$\hat{y}_{T+3}=1$&$\hat{y}_{T+4}=1$&$\hat{y}_{T+5}=0$&$\hat{y}_{T+6}=0$\\
\hline
\end{tabular}
\caption{An example of Phase 1 for $\mathcal{P}_T$ with $T_1=T+6$.}\label{tab.examplephase1}
\end{table}

\begin{table}[!htbp]
\centering
\begin{tabular}{|c|c|c|c|c|c|c|}
\hline
$\hat{x}_{T+6}=1$&$\hat{x}_{T+7}=0$&$\hat{x}_{T+8}=1$&$\hat{x}_{T+9}=0$&$\hat{x}_{T+10}=1$&$\hat{x}_{T+11}=0$&$\hat{x}_{T+12}=1$\\
\hline
$x_{T+6}=1$&$x_{T+7}=0$&$x_{T+8}=1$&$x_{T+9}=0$&$x_{T+10}=1$&$x_{T+11}=0$&$x_{T+12}=1$\\
\hline
$y_{T+6}=1$&$y_{T+7}=0$&$y_{T+8}=1$&$y_{T+9}=0$&$y_{T+10}=1$&$y_{T+11}=0$&$y_{T+12}=1$\\
\hline
$\hat{y}_{T+6}=0$&$\hat{y}_{T+7}=0$&$\hat{y}_{T+8}=1$&$\hat{y}_{T+9}=0$&$\hat{y}_{T+10}=1$&$\hat{y}_{T+11}=0$&$\hat{y}_{T+12}=1$\\
\hline
\end{tabular}
\caption{An example of the first six rounds of Phase 2 for $\mathcal{P}_T$ with $T_1=T+6$.}\label{tab.examplephase2}
\end{table}

Until now, we have proved that $\mathcal{P}_T$ can lead to $\mathcal{E}^{1,1}_{T_2}(u)=1$ where $T_2\leq T+4(c_0+u)+100$.

Then we lower-bound the probability of $\mathcal{P}_T$ happens given $\mathcal{H}_T$ such that $R_{1,T}\geq 0,S_{1,T}\geq -c_0$. Roughly speaking, the probability of each round in $\mathcal{P}$ is lower-bounded by a constant and we know that $\mathcal{P}_T$ has a limited number of rounds, which implies the entire probability of $\mathcal{P}_T$ is bounded by a power of the constant probability lower-bounding a single round in $\mathcal{P}_T$.

For a round $i$ in the first phase, along the prescribed process we have $R_{1,i-1}\geq 0,S_{2,i-1}>0$. By \cref{lemma.nonoptimality34} and the independent sampling convention, the probability of $\{x_i=1-\hat{y}_{i-1}\}\wedge\{y_i=1-\hat{x}_{i-1}\}\wedge\{\text{Alice chooses }\opt_1\}\wedge\{\text{Bob chooses }\opt_2\}$ given the preceding history $\mathcal{H}_{i-1}$ is at least $\min_{a,b\in\{0,1\}}P_{X,Y}(a,b)\times 0.25^2$. For a round $i$ in the second phase, we have $R_{1,i-1}\geq0,S_{1,i-1}\geq0$, and the same lower bound applies to $\{x_i=1-x_{i-1}\}\wedge\{y_i=1-x_{i-1}\}\wedge\{\text{Alice chooses }\opt_1\}\wedge\{\text{Bob chooses }\opt_1\}$. Given $\mathcal{H}_T$, the prescribed path and its completion time $T_2$ are determined by this construction, so the successive conditional bounds can be multiplied; no constraints are imposed after completion.

Therefore, for a history $\mathcal{H}_T$ such that $R_{1,T}\geq 0,S_{1,T}\geq -c_0$, we have\begin{align*}
    &\Pr\left\{\left(\vee_{i=T}^{T+4(u+c_0)+100}\mathcal{E}^{1,1}_i(u)\right)\vee\left(\vee_{i=T}^{T+4(u+c_0)+100}\mathcal{E}^{2,2}_i(u)\right)=1\middle|\mathcal{H}_T\right\}\\
    &\geq\Pr\left\{\mathcal{P}_T\text{ happens}\middle|\mathcal{H}_T\right\}\\
    &\geq \prod_{t=T+1}^{T_2} \left(\min_{i,j\in\{0,1\}}\left\{P_{X,Y}(i,j)\right\}\times0.25^2\right)\\
    &\geq \left(\min_{i,j\in\{0,1\}}\{P_{X,Y}(i,j)\}\times0.25^2\right)^{100+4(u+c_0)}.
\end{align*}
Therefore, we can set $\lambda_u$ as $(\min_{i,j\in\{0,1\}}\{P_{X,Y}(i,j)\}\times0.25^2)^{100+4(u+c_0)}$, which is what we want to prove.
\end{proof}

\subsubsection{Proof of Lemma~\ref{lemma.expectr12increase}}
\label{app.proofexpectr12increase}

\begin{proof}
By the definition of $c_0$, when $\mathcal{E}_t^{1,1}(\lfloor\frac{u}{2}\rfloor)=1$ happens, in each of the next $\zeta+1$ rounds Alice and Bob each choose $\opt_1$ with conditional probability larger than $1-\delta$. The corresponding assertion for $\mathcal{E}_t^{2,2}$ uses $\opt_2$.
Using this, we want to prove that $\mathbb{E}[\sum_{j=1}^{\zeta+1}(r_{1,t+j}-r_{2,t+j})|\mathcal{H}_t\text{ such that }\mathcal{E}_t^{1,1}(\lfloor\frac{u}{2}\rfloor)=1]\geq 100$. Actually, we can deduce that \begin{align}
    &\mathbb{E}\left[\sum_{j=1}^{\zeta+1}(r_{1,t+j}-r_{2,t+j})\middle|\mathcal{H}_t\text{ such that }\mathcal{E}_t^{1,1}\left(\lfloor\frac{u}{2}\rfloor\right)=1\right]\nonumber\\
    &=\sum_{j=1}^{\zeta+1}(\mathbb{E}\left[\mathbb{I}[x_{t+j}=\hat{y}_{t+j}]-\mathbb{I}[1-x_{t+j}=\hat{y}_{t+j}]\middle|\mathcal{H}_t\text{ such that }\mathcal{E}_t^{1,1}\left(\lfloor\frac{u}{2}\rfloor\right)=1\right]\nonumber\\
    &-\mathbb{E}\left[\mathbb{I}[x_{t+j}=\hat{y}_{t+j-1}]-\mathbb{I}[1-x_{t+j}=\hat{y}_{t+j-1}]\middle|\mathcal{H}_t\text{ such that }\mathcal{E}_t^{1,1}\left(\lfloor\frac{u}{2}\rfloor\right)=1\right])\nonumber\\
    &\geq -2+\sum_{j=2}^{\zeta+1}\left(\gamma_1-\gamma_2-4\delta\right)\label{eq.relaxexpectations1}\\
    &\geq -2+\sum_{j=2}^{\zeta+1}\frac{\gamma_1-\gamma_2}{2}\label{eq.removedelta1}\\
    &>100.\nonumber
\end{align}
Here, \cref{eq.relaxexpectations1} follows from the pointwise comparison in the proof of \cref{lemma.delta}, with $i=1$. For $j\geq2$, $\mathbb{E}[A_{t+j}\mid\mathcal{H}_t]=\gamma_1-\gamma_2$ and both deviation indicators have conditional expectations at most $\delta$, by conditioning first on the preceding histories. The first-round term is bounded below by $-2$. Finally, \cref{eq.removedelta1} follows from $4\delta\leq(\gamma_1-\gamma_2)/2$. The same argument applies to Bob, and to option $2$ with its corresponding payoff difference.
\end{proof}

\subsubsection{Proof of Lemma~\ref{lemma.convergecondition}}
\label{app.proofconvergecondition}

\begin{proof}
Without loss of generality, we suppose $\mathcal{H}_T$ satisfies $\mathcal{E}_T^{1,1}(u)$, and what we need is to prove that $\exists u\in\mathbb{N}^+$ such that \begin{align*}
    \Pr\left\{\forall i\in\mathbb{N}, \mathcal{E}_{T+i\left(\zeta+1\right)}^{1,1}\left(\lfloor\frac{u}{2}\rfloor+i\right)=1\middle|\mathcal{H}_T\right\}\geq 1-\varepsilon,
\end{align*}
for any $\varepsilon>0$.

We choose $u$ to be larger than $2c_0+1$ at first. Then let
\begin{align*}
\tau&=\inf\left\{t=T+i(\zeta+1):i\in\mathbb{N}^+,\ R_{1,t}\leq\lfloor u/2\rfloor\text{ or }S_{1,t}\leq\lfloor u/2\rfloor\right\},\\
i_\tau&=(\tau-T)/(\zeta+1),
\end{align*}
where $\inf\varnothing=\infty$. We construct the stopped process
\[
D_i=R_{1,T+(i\wedge i_\tau)(\zeta+1)}-R_{2,T+(i\wedge i_\tau)(\zeta+1)}-100(i\wedge i_\tau),
\]
with filtration $(\mathcal{H}_{T+i(\zeta+1)})_{i\geq0}$.

Now we show that $\{D_i\}_{i=0,1,2,\cdots}$ is a sub-martingale with bounded difference, and we will use Azuma-Hoeffding inequality (\cref{thm:asuma}) to show the value of $D_i$ (and thus $R_{1,t}$) cannot be too small. This implies that $\{D_j\leq u+2j-100j\}$ happens with a probability upper-bounded by a function of $u$. Therefore, the probability of $\{\forall i\in\mathbb{N}, \mathcal{E}_{T+i\left(\zeta+1\right)}^{1,1}\left(\lfloor\frac{u}{2}\rfloor+i\right)=1|\mathcal{H}_T\}$ can be lower-bounded by a function of $u$, which tends towards $1$ when $u$ tends towards infinity.

Actually, we firstly verify that $\{D_i\}_{i\in\mathbb{N}}$ is a sub-martingale. If $\mathcal{H}_{T+(i-1)\left(\zeta+1\right)}$ satisfies that $\mathbb{I}\left[T+(i-1)\left(\zeta+1\right)<\tau\right]=1$, we have $R_{1,T+(i-1)\left(\zeta+1\right)}>\lfloor\frac{u}{2}\rfloor$ and $S_{1,T+(i-1)\left(\zeta+1\right)}>\lfloor\frac{u}{2}\rfloor$. According to \cref{lemma.expectr12increase}, we have \begin{align*}
    &\mathbb{E}\left[D_i\middle|\mathcal{H}_{T+(i-1)\left(\zeta+1\right)}\text{ with }D_{i-1},\cdots,D_0\right]\\
    &=\mathbb{E}\left[\sum_{j=1}^{\zeta+1}\left(r_{1,T+(i-1)\left(\zeta+1\right)+j}-r_{2,T+(i-1)\left(\zeta+1\right)+j}\right)\middle|\mathcal{H}_{T+(i-1)\left(\zeta+1\right)}\right]\\
    &+D_{i-1}-100\\
    &\geq D_{i-1}.
\end{align*}
If $\mathbb{I}\left[T+(i-1)\left(\zeta+1\right)<\tau\right]=0$, we know that $D_{i}=D_{i-1}$.

When $i\leq i_\tau$, we have $|D_i-D_{i-1}|=|2R_{1,T+i\left(\zeta+1\right)}-2R_{1,T+(i-1)\left(\zeta+1\right)}-100|\leq 2(\zeta+1)+100=102+2\zeta$. When $i>i_\tau$, $|D_i-D_{i-1}|=0$.

Since $D_0=2R_{1,T}\geq2u$, according to Azuma-Hoeffding inequality (\cref{thm:asuma}), we have\begin{align}
    \Pr\left\{D_j\leq u+2j-100j\middle|\mathcal{H}_T\right\}\nonumber
    &\leq \Pr\left\{D_j-D_0\leq-98j-u\middle|\mathcal{H}_T\right\}\nonumber\\
    &\leq \exp\left(\frac{-(98j+u)^2}{2j\left(102+2\zeta\right)^2}\right).\nonumber
\end{align}

If $\wedge_{i=0}^{j-1}\mathcal{E}_{T+i\left(\zeta+1\right)}^{1,1}\left(\lfloor\frac{u}{2}\rfloor+i\right)=1$ and $R_{1,T+j\left(\zeta+1\right)}\leq\lfloor\frac{u}{2}+j\rfloor$, we can deduce that $i_\tau\geq j$ and hence $D_j=2R_{1,T+j(\zeta+1)}-100j\leq u+2j-100j$. We have\begin{align*}
    &\Pr\left\{\wedge_{i=0}^{j-1}\mathcal{E}_{T+i\left(\zeta+1\right)}^{1,1}\left(\lfloor\frac{u}{2}\rfloor+i\right)=1,R_{1,T+j\left(\zeta+1\right)}\leq\lfloor\frac{u}{2}\rfloor+j\middle|\mathcal{H}_T\right\}\\
    &\leq \Pr\{D_j\leq u+2j-100j\mid\mathcal{H}_T\}\\
    &\leq \exp\left(\frac{-(98j+u)^2}{2j\left(102+2\zeta\right)^2}\right).
\end{align*}

Symmetrically, we have\begin{align*}
    &\Pr\left\{\wedge_{i=0}^{j-1}\mathcal{E}_{T+i\left(\zeta+1\right)}^{1,1}\left(\lfloor\frac{u}{2}\rfloor+i\right)=1,S_{1,T+j\left(\zeta+1\right)}\leq\lfloor\frac{u}{2}\rfloor+j\middle|\mathcal{H}_T\right\}\\
    &\leq \exp\left(\frac{-(98j+u)^2}{2j\left(102+2\zeta\right)^2}\right).
\end{align*}

Therefore, we have\begin{align*}
    &\Pr\left\{\wedge_{i=0}^{j-1}\mathcal{E}_{T+i\left(\zeta+1\right)}^{1,1}\left(\lfloor\frac{u}{2}\rfloor+i\right)=1,\mathcal{E}_{T+j\left(\zeta+1\right)}^{1,1}\left(\lfloor\frac{u}{2}\rfloor+j\right)=0\middle|\mathcal{H}_T\right\}\\
    &\leq \Pr\left\{\wedge_{i=0}^{j-1}\mathcal{E}_{T+i\left(\zeta+1\right)}^{1,1}\left(\lfloor\frac{u}{2}\rfloor+i\right)=1,R_{1,T+j\left(\zeta+1\right)}\leq\lfloor\frac{u}{2}\rfloor+j\middle|\mathcal{H}_T\right\}\\
    &+\Pr\left\{\wedge_{i=0}^{j-1}\mathcal{E}_{T+i\left(\zeta+1\right)}^{1,1}\left(\lfloor\frac{u}{2}\rfloor+i\right)=1,S_{1,T+j\left(\zeta+1\right)}\leq\lfloor\frac{u}{2}\rfloor+j\middle|\mathcal{H}_T\right\}\\
    &\leq 2\exp\left(\frac{-(98j+u)^2}{2j\left(102+2\zeta\right)^2}\right).
\end{align*}
Furthermore, we can bound the probability of $\Pr\{\forall i\in\mathbb{N}, \mathcal{E}_{T+\left(\zeta+1\right)i}^{1,1}\left(\lfloor\frac{u}{2}\rfloor+i\right)\vee\mathcal{E}_{T+\left(\zeta+1\right)i}^{2,2}\left(\lfloor\frac{u}{2}\rfloor+i\right)=1|\mathcal{H}_T\}$ as\begin{align*}
     &\Pr\left\{\forall i\in\mathbb{N}, \mathcal{E}_{T+\left(\zeta+1\right)i}^{1,1}\left(\lfloor\frac{u}{2}\rfloor+i\right)=1\middle|\mathcal{H}_T\right\}\\
     &\geq 1-\sum_{j=1}^{+\infty}\Pr\left\{\wedge_{i=0}^{j-1}\mathcal{E}_{T+i\left(\zeta+1\right)}^{1,1}\left(\lfloor\frac{u}{2}\rfloor+i\right)=1,\mathcal{E}_{T+j\left(\zeta+1\right)}^{1,1}(\lfloor\frac{u}{2}\rfloor+j)=0\middle|\mathcal{H}_T\right\}\\
     &\geq 1-\sum_{j=1}^{+\infty}2\exp\left(\frac{-(98j+u)^2}{2j\left(102+2\zeta\right)^2}\right)\\
     &\geq 1-\sum_{j=1}^{+\infty}2\exp\left(\frac{-4802j-98u}{\left(102+2\zeta\right)^2}\right)\\
     &=1-\frac{2\exp\left(\frac{-98u}{\left(102+2\zeta\right)^2}\right)}{\exp\left(\frac{4802}{\left(102+2\zeta\right)^2}\right)-1},
\end{align*}
where the subtracted term decays exponentially in $u$. Until now, we can find a sufficient large $u$ such that $\Pr\{\forall i\in\mathbb{N}, \mathcal{E}_{T+\left(\zeta+1\right)i}^{1,1}\left(\lfloor\frac{u}{2}\rfloor+i\right)\vee\mathcal{E}_{T+\left(\zeta+1\right)i}^{2,2}\left(\lfloor\frac{u}{2}\rfloor+i\right)=1|\mathcal{H}_T\}\geq 1-\varepsilon$ for $\mathcal{H}_T$ with $\mathcal{E}_T^{1,1}(u)\vee\mathcal{E}_T^{2,2}(u)=1$.
The same argument applies to $\mathcal{E}^{2,2}_T(u)$ by symmetry. In particular, the proof keeps the same common option at every block time: all the events are $\mathcal{E}^{1,1}$ when the starting history is in $\mathcal{E}^{1,1}_T(u)$, and all are $\mathcal{E}^{2,2}$ in the symmetric case.
\end{proof}

\subsubsection{Proof of Truthful Convergence (Theorem~\ref{thm.converge})}

\label{app.proofconverge}

Before the proof of our final theorem, we introduce a lemma showing that good events happen for infinite many times with probability $1$. Combining \cref{lemma.alwaysnicesituation} in step 2 and \cref{lemma.inflargegap}, we use the martingale theory to get the following lemma.

\begin{lemma}
\label{corr.inflargegap}
Given the game defined in \cref{thm.converge}, for all $u$ we have $\Pr\{\limsup_{t\rightarrow+\infty}\mathcal{E}^{1,1}_t(u)\vee\mathcal{E}^{2,2}_t(u)=1\}=1$.
\end{lemma}

\begin{proof}

Initially, we prove that in order to prove $\Pr\{\limsup_{t\rightarrow+\infty}\mathcal{E}^{1,1}_t(u)\vee\mathcal{E}^{2,2}_t(u)=1\}=1$, we only need to prove \begin{align*}
    \Pr\{\forall t\in\mathbb{N}^+, \mathcal{E}^{1,1}_t(u)\vee\mathcal{E}^{2,2}_t(u)=0\}=0.
\end{align*}

Actually, in order to prove $\Pr\{\limsup_{t\rightarrow+\infty}\mathcal{E}^{1,1}_t(u)\vee\mathcal{E}^{2,2}_t(u)=1\}=1$, we only need to prove that for any $u\in\mathbb{N}^+$, $\Pr\{\exists t\in\mathbb{N}^+, \mathcal{E}^{1,1}_t(u)\vee\mathcal{E}^{2,2}_t(u)=1\}=1$. More specifically, if this claim holds, we can always find a $t_1\in\mathbb{N}^+$ such that $\mathcal{E}^{1,1}_{t_1}(u_1)\vee\mathcal{E}^{2,2}_{t_1}(u_1)=1$ where $u_1>u$, a $t_2>t_1$ such that $\mathcal{E}^{1,1}_{t_2}(u_2)\vee\mathcal{E}^{2,2}_{t_2}(u_2)=1$ where $u_2>\max\{u_1,t_1\}$, a $t_3>t_2$ such that $\mathcal{E}^{1,1}_{t_3}(u_3)\vee\mathcal{E}^{2,2}_{t_3}(u_3)=1$ where $u_3>\max\{u_2,t_2\}$ and so on, choosing each next level above the preceding time because every absolute cumulative payoff through time $t$ is at most $t$. Hence, we can find a sequence $t_1,t_2,\cdots$ such that $\mathcal{E}^{1,1}_{t_i}(u)\vee\mathcal{E}^{2,2}_{t_i}(u)=1$ for every $i\in\mathbb{N}^+$, which indicates that $\Pr\{\limsup_{t\rightarrow+\infty}\mathcal{E}^{1,1}_t(u)\vee\mathcal{E}^{2,2}_t(u)=1\}=1$. In order to prove $\Pr\{\exists t\in\mathbb{N}^+, \mathcal{E}^{1,1}_t(u)\vee\mathcal{E}^{2,2}_t(u)=1\}=1$, we only need to show that $\Pr\{\forall t\in\mathbb{N}^+, \mathcal{E}^{1,1}_t(u)\vee\mathcal{E}^{2,2}_t(u)=0\}=0$.

We know from \cref{lemma.alwaysnicesituation} that with probability $1$, there are infinitely many $t\in\mathbb{N}^+$ such that $\overline{\mathcal{E}_t^{1,2}\vee\mathcal{E}_t^{2,1}}=1$. For any $\mathcal{H}_t$ such that $\overline{\mathcal{E}_t^{1,2}\vee\mathcal{E}_t^{2,1}}=1$, we can find a $\{\mathcal{E}^{1,1}_i(u)\vee\mathcal{E}^{2,2}_i(u)=1\}$ in the next $100+4(c_0+u)$ rounds with probability no less than $\lambda_u$ according to Lemma~\ref{lemma.inflargegap}. Therefore, we can create a sub-martingale \begin{align*}
D_i&=\{\sum_{j=0}^i\sum_{t=T_j}^{T_j+100+4(c_0+u)}\left(\mathcal{E}^{1,1}_t(u)\vee\mathcal{E}^{2,2}_t(u)\right)\text{ where }T_j\text{ is the smallest number such that }\\
&\sum_{t=1}^{T_j}\overline{\mathcal{E}^{1,2}_t\vee\mathcal{E}^{2,1}_t}=(4c_0+4u+101)j+1\}-\lambda_ui.
\end{align*}

The stopping times $T_j$ are finite by \cref{lemma.alwaysnicesituation}; each is an intermediate time, and $T_{j+1}\geq T_j+4c_0+4u+101$, so these windows are disjoint. Use the filtration $\mathcal{F}_i=\mathcal{H}_{T_i+100+4(c_0+u)}$. All earlier window counts are $\mathcal{H}_{T_i}$-measurable. The bound from \cref{lemma.inflargegap} holds conditional on this complete history; the tower property then gives $\mathbb{E}[D_i\mid\mathcal{F}_{i-1}]\geq D_{i-1}$. No independence between different trials is required.

Now we show that $\{D_i\}_{i=0,1,2,\cdots}$ is indeed a sub-martingale with bounded difference, and we will use Azuma-Hoeffding inequality (\cref{thm:asuma}) to show the value of $D_i$ cannot be too small. Therefore, $D_i\leq k-\lambda_ui$ for any $k\in\mathbb{N}^+$ such that $k<i\lambda_u$ can only happen finitely many times by Borel-Cantelli lemma (\cref{thm.borel}). This implies $\limsup_{t\rightarrow+\infty}\mathcal{E}_t^{1,1}(u)\vee\mathcal{E}_t^{2,2}(u)$ happens with probability $1$.

This is because \begin{align*}
    \mathbb{E}[D_{i}|\mathcal{H}_{T_i}]
    &\geq-\lambda_ui+\sum_{j=0}^{i-1}\sum_{t=T_j}^{T_j+100+4(c_0+u)}\left(\mathcal{E}^{1,1}_t(u)\vee\mathcal{E}^{2,2}_t(u)\right)\\
    &+\mathbb{E}\left[\sum_{t=T_i}^{T_{i}+100+4(c_0+u)}\left(\mathcal{E}^{1,1}_t(u)\vee\mathcal{E}^{2,2}_t(u)\right)\middle|\mathcal{H}_{T_i}\right]\\
    &\geq D_{i-1}+\lambda_{u}(i-1)+\lambda_u-\lambda_ui\geq D_{i-1}.
\end{align*}

Moreover, we know that $|D_{i+1}-D_{i}|\leq 4c_0+4u+101$ for $i\in\mathbb{N}^+$, therefore, according to Azuma-Hoeffding inequality (\cref{thm:asuma}), we have\begin{align*}
    \Pr\{D_{i}\leq k-\lambda_ui\}&\leq \Pr\{D_{i}-D_0\leq k-\lambda_ui\}\\
    &\leq \exp\left(-\frac{(\lambda_ui-k)^2}{2i(101+4u+4c_0)^2}\right),
\end{align*}
which is summable in $i$ for every fixed $k$, since it decays exponentially as $i$ increases. This holds for any $k\in\mathbb{N}^+$ and $i>\frac{k}{\lambda_u}$. Therefore, by Borel-Cantelli lemma (\cref{thm.borel}), $\{D_{i}\leq k-\lambda_ui\}$ will happen only finitely often almost surely. If $\{D_{i}\leq k-\lambda_ui\}$ does not happen, we can deduce that $\mathcal{E}^{1,1}_t(u)\vee\mathcal{E}^{2,2}_t(u)$ happens at at least $k$ distinct times, since the windows are disjoint, in the first $T_i+4c_0+4u+100$ rounds. Hence, $\mathcal{E}^{1,1}_t(u)\vee\mathcal{E}^{2,2}_t(u)$ happens with probability $1$ and consequentially, $\Pr\{\limsup_{t\rightarrow+\infty}\mathcal{E}^{1,1}_t(u)\vee\mathcal{E}^{2,2}_t(u)=1\}=1$.
\end{proof}

Now combining \cref{corr.inflargegap} and \cref{lemma.convergecondition}, we complete our final proof (\cref{thm.converge}). 

\begin{proof}

For any $u\in\mathbb{N}^+$, put $L=\lceil1000/(\gamma_1-\gamma_2)\rceil+1$ and $a=\lfloor u/2\rfloor$. Write $R_s=(R_{1,s},\ldots,R_{4,s})$, and similarly $S_s$, and let $d_i(V)=V_i-\max_{j\neq i}V_j$. Denote by $\mathcal{C}$ the event that, for some common $J\in\{1,2\}$ and finite $B$, both $d_J(R_s)$ and $d_J(S_s)$ are at least $s/L-B$ for every sufficiently large $s$.

Consider the original growing event
\[
    \mathcal{E}^{\text{Convergence Condition}}_T
    =\left\{\forall i\geq0:\;
    \mathcal{E}_{T+Li}^{1,1}(a+i)\vee\mathcal{E}_{T+Li}^{2,2}(a+i)=1\right\}.
\]
Each accumulated payoff changes by at most $L$ in a block. By \cref{lemma.sumofpayoffs}, a switch between its two alternatives at consecutive block times requires a change of at least $2a+2i+1$, which exceeds $L$ for all sufficiently large $i$. Thus the common option with the largest accumulated payoff is eventually a fixed $J$. For $T+Li\leq s<T+L(i+1)$, the accumulated payoffs of this option are at least $a+i-L$. By \cref{lemma.boundeduninformative,lemma.sumofpayoffs}, their differences from all other options are at least $a+i-L-1$ for sufficiently large $i$. Since $i\geq(s-T)/L-1$, this proves $\mathcal{E}^{\text{Convergence Condition}}_T\subseteq\mathcal{C}$.

Suppose $\Pr\{\mathcal{C}=1\} = 1-\varepsilon_2< 1$.
Let $\varepsilon_1=\frac{1}{2}\varepsilon_2$ and $u_1$ satisfy \cref{lemma.convergecondition}, and use $u=u_1$ in the growing event below. Then we can deduce that
\begin{align*}
    &\Pr\left\{\mathcal{C}\right\}\\
    &\geq\sum_{T=1}^{+\infty}\Pr\left\{\mathcal{C},\vee_{t=1}^{T-1}\mathcal{E}_t^{1,1}(u_1)\vee\mathcal{E}_t^{2,2}(u_1)=0,\mathcal{E}_T^{1,1}(u_1)\vee\mathcal{E}_T^{2,2}(u_1)=1\right\}\\
    &=\sum_{T=1}^{+\infty}\Pr\left\{\mathcal{C}\middle|\left(\vee_{t=1}^{T-1}\mathcal{E}_t^{1,1}(u_1)\vee\mathcal{E}_t^{2,2}(u_1)=0\right)\wedge\left(\mathcal{E}_T^{1,1}(u_1)\vee\mathcal{E}_T^{2,2}(u_1)=1\right)\right\}\\
    &\times \Pr\left\{\vee_{t=1}^{T-1}\mathcal{E}_t^{1,1}(u_1)\vee\mathcal{E}_t^{2,2}(u_1)=0,\mathcal{E}_T^{1,1}(u_1)\vee\mathcal{E}_T^{2,2}(u_1)=1\right\}\\
    &\geq \sum_{T=1}^{+\infty}\Pr\left\{\mathcal{E}^{\text{Convergence Condition}}_T\middle|\left(\vee_{t=1}^{T-1}\mathcal{E}_t^{1,1}(u_1)\vee\mathcal{E}_t^{2,2}(u_1)=0\right)\wedge\left(\mathcal{E}_T^{1,1}(u_1)\vee\mathcal{E}_T^{2,2}(u_1)=1\right)\right\}\\
    &\times \Pr\left\{\vee_{t=1}^{T-1}\mathcal{E}_t^{1,1}(u_1)\vee\mathcal{E}_t^{2,2}(u_1)=0,\mathcal{E}_T^{1,1}(u_1)\vee\mathcal{E}_T^{2,2}(u_1)=1\right\}\\
    &\geq \sum_{T=1}^{+\infty}\inf_{T'\in\mathbb{N}^+,\mathcal{H}_{T'}\in \mathcal{E}_{T'}^{1,1}(u_1)\vee\mathcal{E}_{T'}^{2,2}(u_1)}\Pr\left\{\mathcal{E}^{\text{Convergence Condition}}_{T'}\middle|\mathcal{H}_{T'}\right\}\\
    &\times \Pr\left\{\vee_{t=1}^{T-1}\mathcal{E}_t^{1,1}(u_1)\vee\mathcal{E}_t^{2,2}(u_1)=0,\mathcal{E}_T^{1,1}(u_1)\vee\mathcal{E}_T^{2,2}(u_1)=1\right\}\\
    &= \inf_{T\in\mathbb{N}^+,\mathcal{H}_T\in \mathcal{E}_T^{1,1}(u_1)\vee\mathcal{E}_T^{2,2}(u_1)}\Pr\left\{\mathcal{E}^{\text{Convergence Condition}}_T\middle|\mathcal{H}_T\right\}\Pr\left\{\vee_{t=1}^{+\infty}\mathcal{E}_t^{1,1}(u_1)\vee\mathcal{E}_t^{2,2}(u_1)=1\right\}
    \\
    &=\inf_{T\in\mathbb{N}^+,\mathcal{H}_T\in \mathcal{E}_T^{1,1}(u_1)\vee\mathcal{E}_T^{2,2}(u_1)}\Pr\left\{\mathcal{E}^{\text{Convergence Condition}}_T\middle|\mathcal{H}_T\right\}\tag{by \cref{corr.inflargegap}}\\
    &\geq 1-\varepsilon_1,\tag{by \cref{lemma.convergecondition}}
\end{align*}
which is a contradiction.

Hence $\Pr(\mathcal{C})=1$. To pass to realized pure strategies, define the history-measurable indices
\[
    \ell_t^X=\min\argmax_{i\in[4]}R_{i,t-1},\qquad
    A_t^X=\{\opt_t^X\neq\opt_{\ell_t^X}\},
\]
and similarly for Bob. On $\mathcal{C}$ the indices eventually equal the same $J$. With the deterministic $c=1/(2L)>0$, both accumulated payoff differences are at least $ct$ for all sufficiently large $t$. The added condition gives, eventually almost surely,
\[
    \Pr(A_t^X\mid\mathcal{H}_{t-1})\leq q_f(ct),\qquad
    \Pr(A_t^Y\mid\mathcal{H}_{t-1})\leq q_g(ct).
\]
Monotonicity and integrability imply $\sum_tq_k(ct)<\infty$, since $q_k(ct)\leq c^{-1}\int_{c(t-1)}^{ct}q_k(x)\,dx$ for all sufficiently large $t$. Conditional Borel--Cantelli therefore shows that both agents choose other options only finitely often almost surely. Their realized pure strategies eventually equal the common $\opt_J$, which is \cref{def.truthconvergence}.
\end{proof}

\section{Error Bars of Simulation in Section~\ref{sec.simu}}
\label{app.errorbar}

The original simulations were reported as running single-threadedly on AMD Ryzen\texttrademark~R7-5700U Processor at 3.60GHz with 16GB DDR4 SDRAM. The total running time of all the simulations in \cref{Fig.convrate} is $421$ seconds. Code is available in \url{https://github.com/fengtony686/peer-prediction-convergence}. We give error bars of converge rates of each algorithm shown in \cref{Fig.convrate} one by one. Each error bar is drawn by computing converge proportion of each algorithm in $400$ repeating simulations for $10$ times. In the linked code, the center line is the mean of these ten batches, and the shaded region is their pointwise minimum-to-maximum range.

\begin{figure}[H] 
\centering
\includegraphics[width=0.99\textwidth]{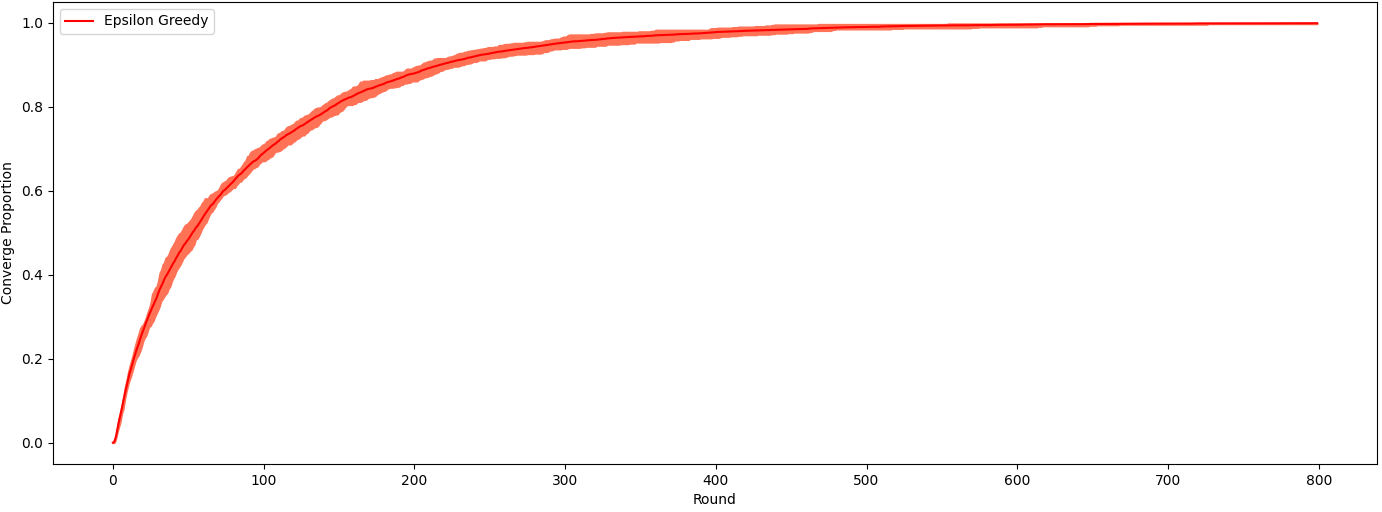}
\caption{Error Bar of $\epsilon$-Greedy in Section~\ref{sec.simu}.}
\label{Fig.errepsgreedy}
\end{figure}

\begin{figure}[H] 
\centering
\includegraphics[width=0.99\textwidth]{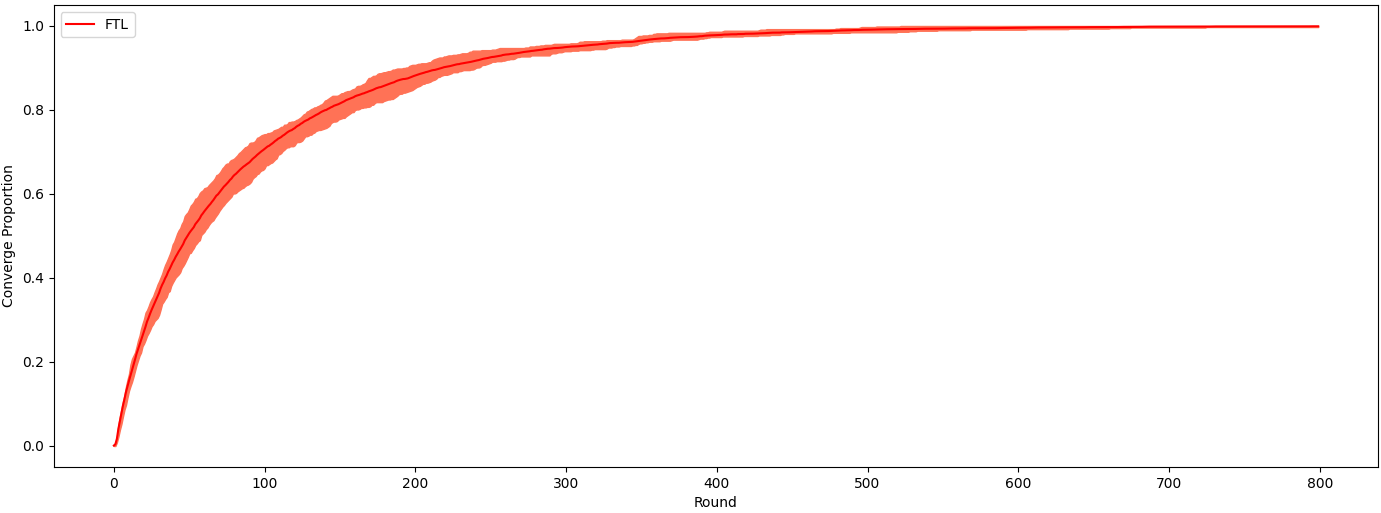}
\caption{Error Bar of FTL in Section~\ref{sec.simu}.}
\label{Fig.errftl}
\end{figure}

\begin{figure}[H] 
\centering
\includegraphics[width=0.99\textwidth]{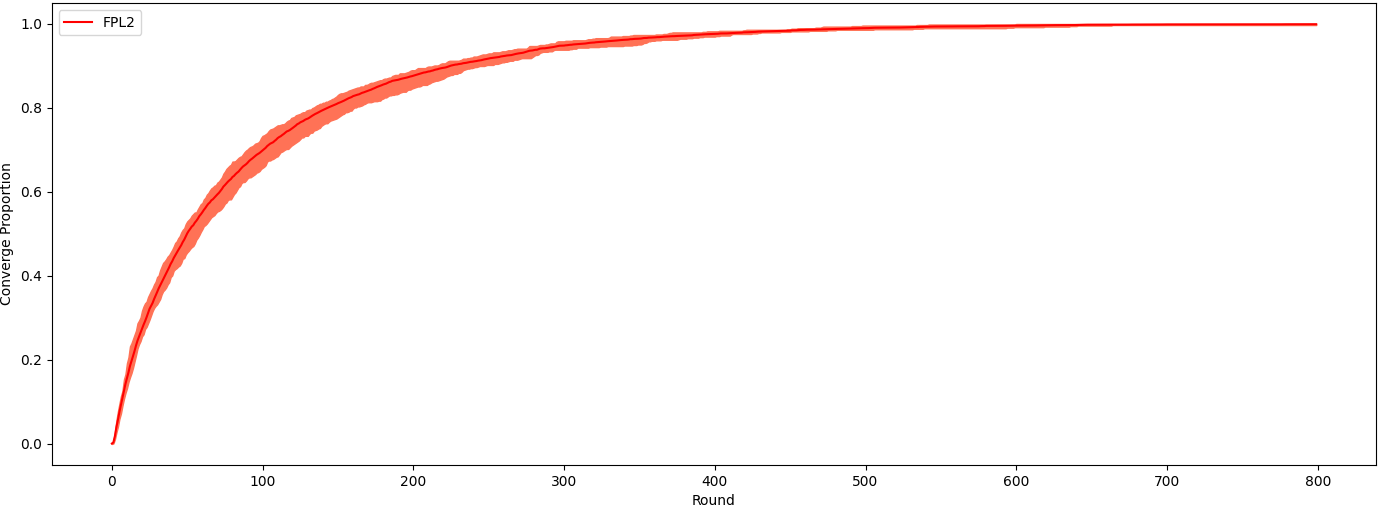}
\caption{Error Bar of FPL2 in Section~\ref{sec.simu}.}
\label{Fig.errfpl2}
\end{figure}

\begin{figure}[H] 
\centering
\includegraphics[width=0.99\textwidth]{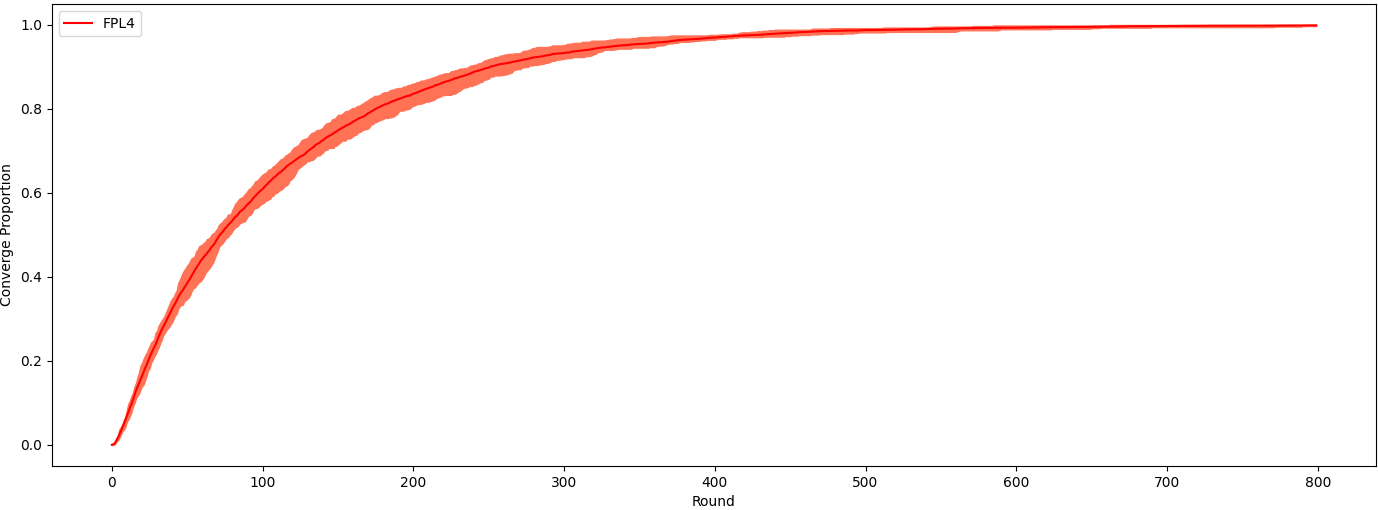}
\caption{Error Bar of FPL4 in Section~\ref{sec.simu}.}
\label{Fig.errfpl4}
\end{figure}

\begin{figure}[H] 
\centering
\includegraphics[width=0.99\textwidth]{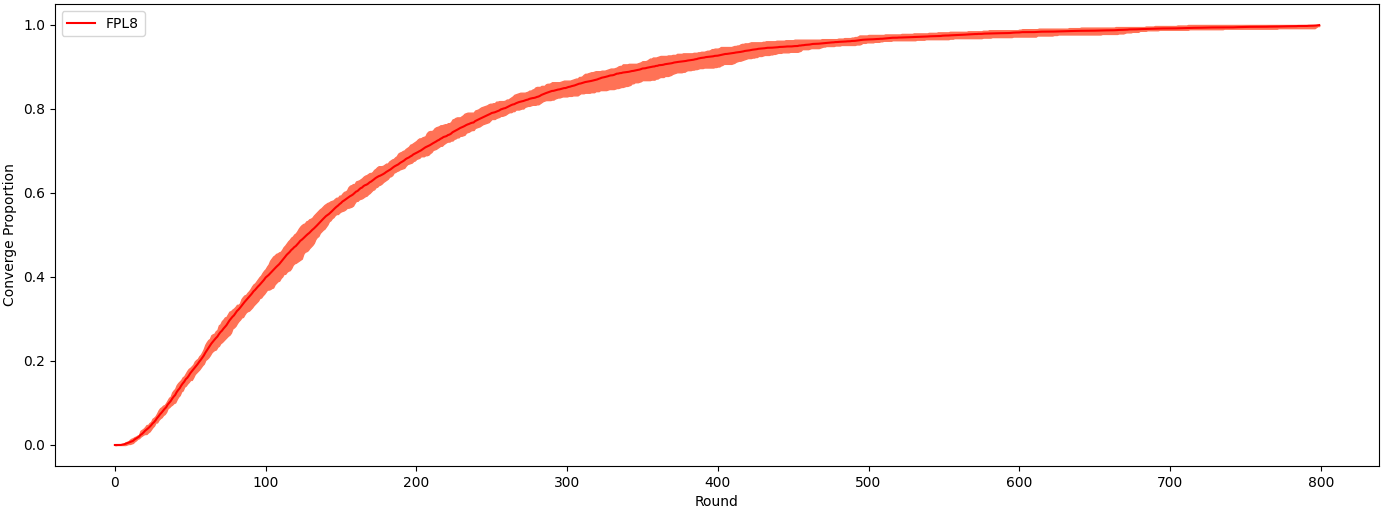}
\caption{Error Bar of FPL8 in Section~\ref{sec.simu}.}
\label{Fig.errfpl8}
\end{figure}

\begin{figure}[H] 
\centering
\includegraphics[width=0.99\textwidth]{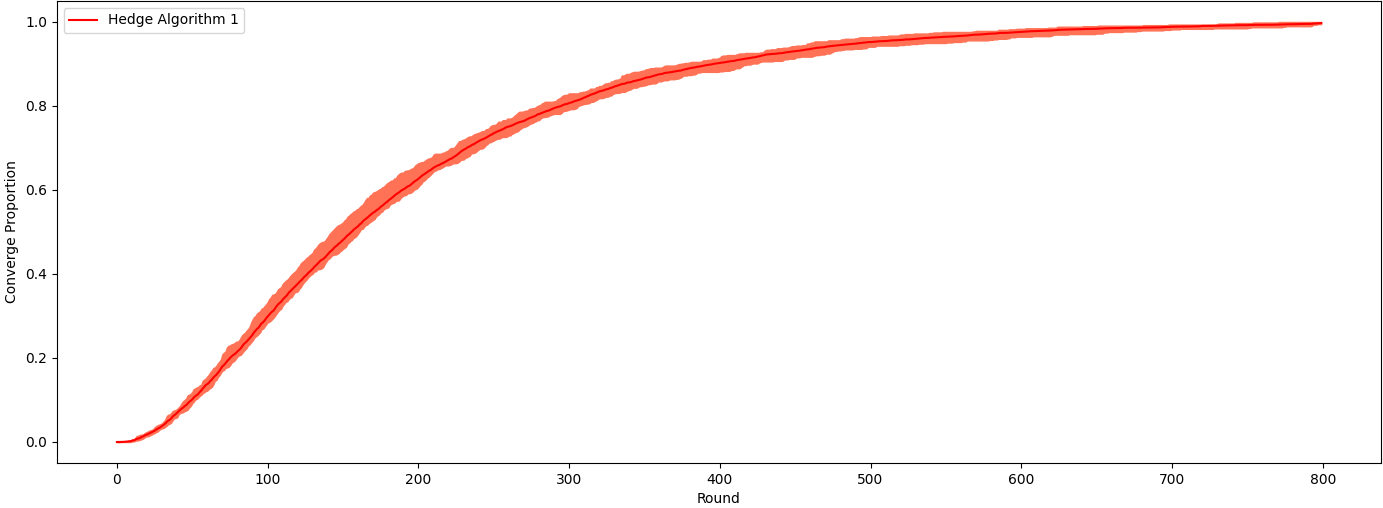}
\caption{Error Bar of Hedge Algorithm 1 in Section~\ref{sec.simu}.}
\label{Fig.errhedge1}
\end{figure}

\begin{figure}[H] 
\centering
\includegraphics[width=0.99\textwidth]{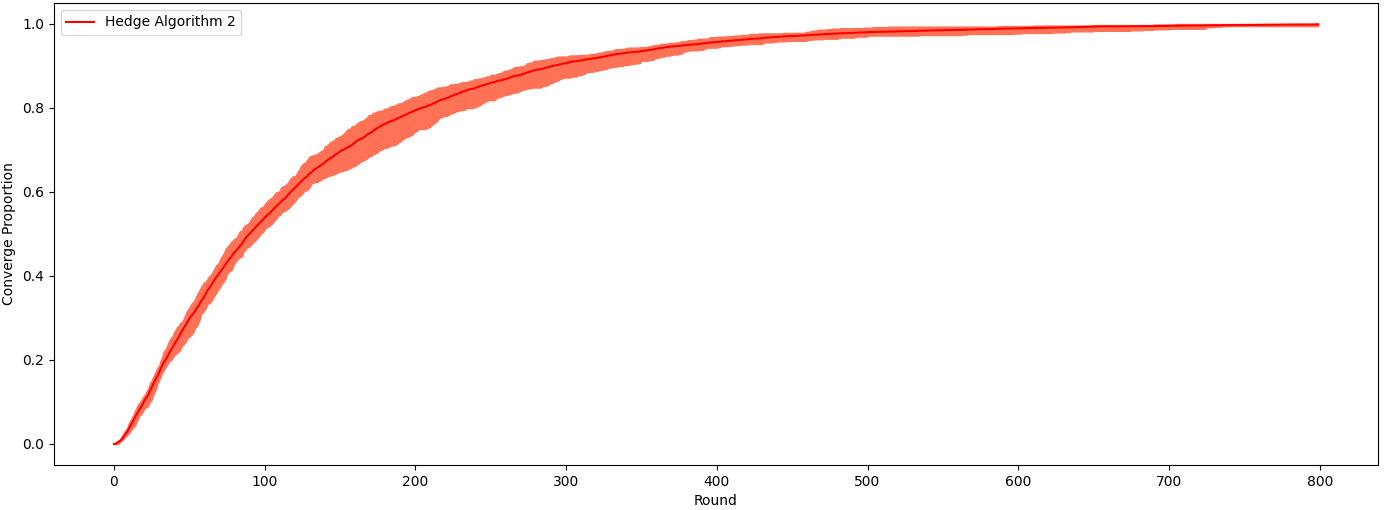}
\caption{Error Bar of Hedge Algorithm 2 in Section~\ref{sec.simu}.}
\label{Fig.errhedge2}
\end{figure}

\end{document}